\documentclass[hidelinks,onefignum,onetabnum,a4paper]{siamart251216}
\usepackage{amsmath, amssymb}
\usepackage{geometry}
\usepackage{graphicx}
\usepackage{hyperref}
\usepackage{multirow}

\newtheorem{assumption}{Assumption}
\newtheorem{remark}{Remark}

\usepackage{algorithm}
\usepackage{algpseudocode}

\begin{document}

\title{An Asynchronous Multi-Rate Taylor Method for Delay Differential
  Equations}

\author{Avinash Malik \thanks{Department of Electrical, Computer, and
    Software Engineering, University of Auckland, NZ
    (avinash.malik@auckland.ac.nz)}} \date{\today}

\maketitle

\begin{abstract}
  The numerical simulation of high-dimensional, multi-rate Delay
  Differential Equations (DDEs) is fundamentally bottlenecked by
  synchronous time-stepping and the dynamic memory allocation required
  for continuous history tracking. In this paper, we introduce the
  Asynchronous Adaptive Taylor Solver (AATS), an event-driven
  integration framework designed to overcome these high-performance
  computing limitations. By assigning independent local clocks to
  individual coordinates and advancing them using high-order Taylor
  polynomials generated via compile-time Automatic Differentiation, AATS
  restricts computational work to actively evolving sub-graphs. To
  eliminate the severe memory overhead endemic to traditional DDE
  solvers, AATS utilizes statically allocated circular buffers to store
  polynomial segments, achieving interpolation-free continuous
  dense-output evaluation with a verified zero-allocation runtime memory
  footprint.

  Alongside this software architecture, we establish a novel continuous
  proof of convergence for asynchronous Taylor expansions and formally
  prove that the framework's algorithmic complexity scales linearly
  ($\mathcal{O}(N)$). Extensive benchmarks against state-of-the-art synchronous
  solvers (Julia SciML) validate these theoretical bounds. On
  large-scale benchmarks (upto $N = 10000$ coordinates) AATS
  fundamentally minimizes the constant factor of algorithmic work by
  avoiding redundant evaluations, delivering empirically consistent with
  $\mathcal{O}(N)$ execution scaling and significant wall-clock speedups.
\end{abstract}


\begin{keywords}
  Delay differential equations, asynchronous integration, multi-rate methods, Taylor series, automatic differentiation, zero-allocation memory
\end{keywords}

\begin{MSCcodes}
  65L03, 65L20, 65Y20, 65L05
\end{MSCcodes}


\section{Introduction}
\label{sec:intro}

The numerical simulation of Delay Differential Equations (DDEs) is a
cornerstone of modern computational science \cite{bellen2013numerical, hale1993introduction},
governing the modeling of systems where the rate of change depends not
only on the current state but also on the system's history. These
equations naturally arise in fields characterized by inherent
transmission or reaction lags, ranging from biological neural networks \cite{campbell1999delay}
and population dynamics \cite{kuang1993delay}
to control theory \cite{richard2003time}
and epidemiological modeling \cite{bocharov2000numerical}.
However, as the dimensionality of these network topologies scales to
tens of thousands of coupled variables, simulating DDEs introduces
profound computational bottlenecks that stretch the limits of
traditional numerical integrators.

The first fundamental bottleneck is temporal. State-of-the-art numerical
frameworks, like Python's SciPy~\cite{virtanen2020scipy} or Julia's
SciML~\cite{rackauckas2017differentialequations}, almost exclusively
rely on synchronous time-stepping. In a synchronous paradigm, the global
integration step size $\Delta t$ is bounded by the fastest-evolving
coordinate in the network~\cite{hairer1996solving, gear1984multirate}.
For highly heterogeneous, multi-rate systems, where a few localized
variables may exhibit explosive transient dynamics while the vast
majority remain dormant. This shared global grid forces the solver to
evaluate the entire $N$-dimensional system at a microscopic frequency.
This leads to an overwhelming number of mathematically redundant
evaluations.

The second, equally severe bottleneck is infrastructural. Unlike
Ordinary Differential Equations (ODEs), which only require the current
state vector to advance time, DDEs require continuous evaluation of
historical states $x(t-\tau(t))$. To achieve this, traditional solvers
typically maintain two parallel data structures: a discrete integrator
to advance the state, and a secondary dense-output interpolant (such as
Hermite splines) to query the history
continuously~\cite{shampine2001solving, bellen2013numerical}. In modern
operating systems, storing and updating this history dynamically
requires continuous calls to the heap manager. Over many integration
steps, this dynamic memory allocation leads to catastrophic memory
fragmentation, cache misses, and garbage collection
latency~\cite{drepper2007what, hager2010introduction}.

To overcome these intertwined mathematical and computational
limitations, we introduce the Asynchronous Adaptive Taylor Solver
(AATS). AATS is a high-performance C++ framework designed to completely
abandon the shared global time grid and the reliance on dynamic memory
allocation. By assigning each spatial coordinate its own independent
local time domain, and by advancing these coordinates using high-order
Taylor polynomials generated via compile-time Automatic Differentiation
(AD) \cite{griewank2008evaluating}, AATS ensures that computational work
is spent only where the local error dynamics demand it.

The primary contributions of this paper are fourfold:

\begin{enumerate}
\item \textbf{Asynchronous Temporal Decoupling:} We present
  \textit{event-driven} integration, which entirely decouples the
  temporal evolution of multi-rate components. By employing an amortized
  $\mathcal{O}(1)$ priority queue, the method bypasses the global step-size
  constraints imposed by rapidly changing system dynamics, ensuring that
  the global computational complexity scales with localized dynamics
  rather than worst-case properties.
\item \textbf{Zero-Allocation History Architecture:} We introduce a
  static memory framework for continuous history tracking. By storing
  fixed-degree Taylor segments within statically allocated circular
  buffers, the algorithm formally eliminates runtime dynamic memory
  allocation. This approach mathematically bounds the peak memory
  footprint and optimizes spatial cache locality without sacrificing
  historical continuity.
\item \textbf{Interpolation-Free Delay Evaluation:} We establish that
  the local high-order Taylor polynomials generated during integration
  inherently function as exact continuous state representations. By
  directly evaluating these polynomials, the solver evaluates continuous
  and state-dependent delays analytically, eliminating the computational
  and memory overhead associated with constructing secondary
  dense-output interpolants (e.g., Hermite splines).
\item \textbf{Unified Convergence and Complexity Theory:} We establish a
  comprehensive mathematical framework that couples the asynchronous
  convergence of DDE networks with their theoretical execution cost. By
  utilizing continuous extensions of generated Taylor polynomials, we
  derive strict global error bounds and delay-differential stability
  conditions. Furthermore, we leverage these error bounds to formally
  prove that the total algorithmic work scales exclusively with the
  localized intrinsic dynamics as a function of the user-defined
  tolerance $\varepsilon$, breaking the global worst-case complexity barrier.
\end{enumerate}

The remainder of this paper is organized as follows. Section
\ref{sec:example} introduces a motivating $10,000$-dimensional
Continuous-Time Recurrent Neural Network (CTRNN) to contextualize the
multi-rate problem. Section \ref{sec:algorithm} details the asynchronous
algorithmic architecture and data structures. Sections
\ref{sec:complexity} and \ref{sec:convergence} formalize the
mathematical framework, proving the algorithmic complexity bounds and
the continuous asynchronous convergence guarantees, respectively.
Section \ref{sec:results} benchmarks AATS against the state-of-the-art
Julia SciML framework, demonstrating allocation reductions and
significant wall-clock speedups across varying hardware architectures.
Section~\ref{sec:related_work} compares AATS with current
state-of-the-art. The paper concludes in Section~\ref{sec:conclusion}.

\section{A Motivating Example}
\label{sec:example}

To illustrate the severe computational bottlenecks inherent to traditional synchronous DDE integration, consider the simulation of a large-scale Continuous-Time Recurrent Neural Network (CTRNN) with sparse local connectivity and a uniform axonal transmission delay $\tau=0.1$. 

Let the network consist of $N = 10,000$ neurons, where
$U(t) \in \mathbb{R}^N$ is the vector of membrane potentials. The first component
$u_1(t)$ represents a high-frequency sensory input neuron, while the
remaining components $u_2(t) \dots u_N(t)$ represent a large, slowly
adapting hidden reservoir. The network dynamics are governed by the
standard matrix DDE:
$$
\dot{U}(t) = -A \, U(t) + W \tanh(U(t-\tau)),
$$

\noindent
where the decay matrix $A$ and the synaptic weight matrix $W$ are
defined as:
$$
A = \underbrace{\begin{bmatrix}
10^4 & 0 & \cdots & 0 \\
0 & 1 & \cdots & 0 \\
\vdots & \vdots & \ddots & \vdots \\
0 & 0 & \cdots & 1
\end{bmatrix}}_{\text{Rapidly changing Decay Matrix}},
\quad
W = \underbrace{\begin{bmatrix}
0 & 10^4 & 0 & \cdots & 0 \\
0.5 & 0 & 0.5 & \cdots & 0 \\
\vdots & \ddots & \ddots & \ddots & \vdots \\
0 & \cdots & 0.5 & 0 & 0.5
\end{bmatrix}}_{\text{Sparse Tridiagonal Weight Matrix}}.
$$

This architecture mathematically isolates two defining computational
properties: (1) \textbf{Multi-Rate changes:} The spectral gap in $A$
dictates that the sensory neuron processes inputs on a microscopic
timescale ($\approx 10^{-4}$), while the reservoir neurons evolve on a
macroscopic timescale ($\approx 1.0$). (2) \textbf{Sparsity:} To mimic
biological topology, $W$ is strictly sparse (specifically, a hollow
tridiagonal structure). The local synaptic fan-in for any given neuron
is explicitly bounded by $E_v \le 2$.

\subsection{The Synchronous Bottleneck}
If this network is simulated using a state-of-the-art synchronous DDE
solver, the step-size controller must enforce a single global
integration step $\Delta t$. To stably resolve the extreme eigenvalue of the
sensory neuron, the solver is violently constrained to a global
micro-step of $\Delta t \approx 10^{-5}$.

Over a brief $1$-second simulation window, the synchronous solver
executes $100,000$ discrete time steps. At \textit{every single
  micro-step}, it must evaluate the full dense vector field for all
$10,000$ neurons. This results in $10^9$ redundant \texttt{tanh}
evaluations and large memory allocations to dynamically store the
network's delayed history, despite $99.99\%$ of the reservoir undergoing
virtually no physical change.

\subsection{A Proposed Coordinate-Wise Asynchronous Solution}
\label{sec:prop-coord-wise}

To resolve this synchronization bottleneck, we shift to coordinate-wise
asynchrony. Each variable $u_i(t)$ is assigned an independent local
clock and maintains a continuous mathematical history. Fast components
evolve via independent micro-steps, while slow variables take
macroscopic steps and remain dormant.

When an active variable requires delayed information (e.g.,
$u_2(t-\tau)$), it simply evaluates its neighbor's dormant continuous
trajectory algebraically, avoiding global synchronization. Exploiting
network sparsity, rapid updates only trigger localized preemption checks
on direct neighbors, waking them early only if their error thresholds
are violated.

By decoupling temporal domains, this ``sleep and wake'' formulation
reduces computational workload by orders of magnitude (e.g., from $10^9$
global evaluations down to $200,000$ localized evaluations in our
example). Section \ref{sec:algorithm} details the algorithmic machinery
and zero-allocation memory structures enabling this scheme.

\section{The Asynchronous Adaptive Taylor Solver (AATS)}
\label{sec:algorithm}

\begin{algorithm}[h]
  \caption{Asynchronous Adaptive Taylor Solver (AATS)}\label{alg:aats}
  \small
  \begin{algorithmic}[1]
    \Require Initial state $x_0 \in \mathbb{R}^N$, Final time $T$, Local tolerance $\varepsilon$, Polynomial degree $K$
    \Ensure Continuous history polynomials for all coordinates $v \in \{1, \dots, N\}$ up to $T$
    \Procedure{AATS\_Solve}{$x_0, T, \varepsilon$}
      \State $\mathcal{H} \gets \text{Empty RadixHeap}$ \Comment{Amortized $\mathcal{O}(1)$ coordinate event scheduler}
      \State $\mathcal{B} \gets \text{Array of } N \text{ Zero-Allocation RingBuffers}$
    
    \Statex \Comment{\textit{Phase 1: Coordinate-Wise Initialization}}
    \For{each coordinate $v \in \{1, \dots, N\}$}
    \State $c^{(v)} \gets \text{ForwardAD}(v, 0, x_0)$ \Comment{Compile-time AD graph for coordinate $v$}
    \State $\Delta t \gets \text{StepSizeController}(c^{(v)}_{K}, \varepsilon)$
    \State $\mathcal{B}[v]\text{.push}(c^{(v)}, \text{time}=0)$
    \State $\mathcal{H}\text{.push}(\text{time}=\Delta t, \text{id}=v)$
    \EndFor

    \Statex \Comment{\textit{Phase 2: Main Asynchronous Coordinate Event Loop}}
    \While{$\mathcal{H} \text{ is not empty}$}
    \State $(t, v) \gets \mathcal{H}\text{.pop}()$ \Comment{Extract next active coordinate $v$}
    \If{$t \geq T$}
    \State \textbf{break}
    \EndIf
    
    \Statex \indent\Comment{\textit{Native Update for Active Coordinate $v$}}
    \State $x_v(t) \gets \text{EvaluateHistory}(\mathcal{B}[v], t)$
    \State $c^{(v)} \gets \text{ForwardAD}(v, t, x_v(t))$ \Comment{$\mathcal{O}(K)$ evaluation of coordinate $v$'s AD graph}
    \State $\Delta t \gets \text{StepSizeController}(c^{(v)}_{K}, \varepsilon)$
    \State $\mathcal{B}[v]\text{.push}(c^{(v)}, \text{time}=t)$ \Comment{$\mathcal{O}(1)$ Static overwrite of coordinate history}
    \State $\mathcal{H}\text{.push}(\text{time}=t + \Delta t, \text{id}=v)$ \Comment{Schedule next update for coordinate $v$}
    
    \Statex \indent\Comment{\textit{Dependency Resolution: Preempt Dormant Coupled Coordinates}}
    \For{each coupled coordinate $j \in \text{Dependencies}(v)$}
    \State $t_j \gets \mathcal{B}[j]\text{.latest\_time}()$
    \State $\tilde{c}^{(j)} \gets \text{BinomialShift}(\mathcal{B}[j]\text{.latest\_coeffs}(), t - t_j)$
    \Statex \Comment{True derivative of coordinate $j$ from evolution function}
    \State $d_{\text{true}} \gets \text{EvaluateDynamics}(j, x_v(t))$
    \State $d_{\text{pred}} \gets \tilde{c}^{(j)}_1$ \Comment{Predicted derivative of coordinate $j$}
    
    \If{$|d_{\text{true}} - d_{\text{pred}}| > \varepsilon$}
    \Statex \Comment{Preempt coordinate $j$: Schedule immediate update}
    \State $\mathcal{H}\text{.push}(\text{time}=t, \text{id}=j)$ 
    \EndIf
    \EndFor
    \EndWhile
    \EndProcedure
  \end{algorithmic}
\end{algorithm}

Algorithm \ref{alg:aats} gives the pseudocode of the Asynchronous
Adaptive Taylor Solver (AATS). At its core, the algorithm operates by
replacing the traditional global time-step loop with an event-driven
priority queue (line 2). Instead of storing discrete historical data
points requiring dynamic memory allocation, AATS models the history of
every coordinate\footnote{In the context of network or multi-dimensional
  differential equations, each node, vertex, or individual state
  variable maps directly to a distinct spatial coordinate in the
  system's state space $\mathbb{R}^N$.} as a continuous, piece-wise Taylor
polynomial. These polynomials are statically allocated in ring buffers
(line 3), strictly enforcing a zero-allocation memory footprint during
runtime. The solver advances by integrating the most temporally imminent
coordinate, analytically extracting delayed states from dormant
neighbors, and selectively preempting those neighbors if the active
coordinate's trajectory violates their local error bounds.

\subsection{Algorithmic Walkthrough of the CTRNN Simulation}
\label{sec:algo_walkthrough}

We map the high-level logic of Algorithm \ref{alg:aats} to the $10,000$-neuron CTRNN (Section \ref{sec:example}), tracking an update cycle for the fast sensory neuron ($u_1$).

\paragraph{Phase 1: Initialization (Lines 4--9)}
At $t=0$, the algorithm initializes continuous trajectories (Line 5),
using Automatic Differentiation (AD), and evaluates local truncation
errors (Line 6). Detecting the network's multi-rate dynamics, the
controller assigns the fast neuron $u_1$ a micro-step
($\Delta t_1 = 10^{-5}$) while the $9,999$ slow reservoir neurons receive
macroscopic steps ($\Delta t_i \approx 0.1$). Variables are queued chronologically
(Lines 7--8).

\paragraph{Phase 2: Asynchronous Event Loop (Lines 11--19)}
Assuming the simulation reaches $t=1.0$, the scheduler pops the most
imminent event: $u_1$ (Line 11). To evaluate $u_1$'s nonlinear dynamics,
it requires the delayed state of its neighbor, $\tanh(u_2(t-\tau))$.
Instead of forcing $u_2$ to update, the solver directly evaluates
$u_2$'s stored continuous history (Lines 15--16). Neuron $u_1$ then
generates its next polynomial, takes its micro-step, and reschedules
(Lines 17--19).

\paragraph{Phase 3: Sparse Dependency Resolution and Preemption (Lines 20--28)}
Because $u_1$'s state changed, it may invalidate the planned trajectories of its neighbors. Exploiting sparsity (Line 20), the algorithm only checks $u_2$ (its single direct connection), safely ignoring the remaining $9,998$ dormant neurons. It aligns $u_2$'s dormant trajectory to the current time (Lines 21--24) and tests the true derivative error. If the tolerance is exceeded, $u_2$ is preempted (Lines 25--27) and forced to update immediately, guaranteeing global accuracy without synchronous barriers.

\vspace{0.5em}

Now that we have described the working on the algorithm, we describe the
high-performance techniques used to overcome the inefficiencies
introduced by asynchronous integration.

\subsection{Compile-Time AD Graphs}
\label{sec:compile_time_ad}
Standard runtime Automatic Differentiation (AD) incurs prohibitive tape-building overhead. AATS parses dependencies at compile time into a static, allocation-free graph. Fixing the polynomial degree $K$ unrolls derivative expansions into Fused Multiply-Add (FMA) instructions, computing coefficients (Algorithm \ref{alg:aats}, Line 16) in pure $\mathcal{O}(K)$ time entirely within CPU registers.

\subsection{Amortized $\mathcal{O}(1)$ Event Scheduling}
\label{sec:radix_heap}
To process independent coordinate updates, AATS uses a Radix Heap priority queue (Algorithm \ref{alg:aats}, Lines 11, 19). Exploiting strictly monotonic time advancement reduces event extraction to an amortized $\mathcal{O}(1)$. Preempted dormant coordinates (Line 26) are handled via lazy deletion, allowing the queue to silently discard obsolete updates in $\mathcal{O}(1)$ time without search-and-remove overhead.

\subsection{Zero-Allocation Memory Management}
\label{sec:zero_allocation}
Dynamic heap allocation for history tracking causes severe latency. AATS
guarantees zero runtime allocations using a static array of $N$ circular
buffers, each storing $M$ fixed-degree polynomial segments. New
coefficients are pushed directly to these buffers (Algorithm
\ref{alg:aats}, Line 18), natively overwriting old entries. This
strictly bounds memory to $\mathcal{O}(N \cdot M \cdot K)$ and maximizes L1 cache
locality.

\subsection{Binomial Time-Alignment}
\label{sec:binomial_shifting}
Evaluating coupled dynamics requires aligning dormant neighbor states at historical times ($t_j$) to the active time ($t$). Rather than triggering expensive re-evaluations, AATS projects historical states forward algebraically using a binomial expansion (Algorithm \ref{alg:aats}, Line 22). This $\mathcal{O}(K^2)$ shift executes purely as register instructions, synchronizing dependencies without breaking the CPU pipeline.

\subsection{Interpolation-Free Delay Lookup}
\label{sec:delay_lookup}
Continuous delays typically require expensive secondary interpolants.
AATS eliminates this because the integration Taylor series inherently
serves as an exact dense-output representation. To evaluate a delayed
state $x(t-\tau(x(t), t))$, AATS locates the relevant buffer segment and
directly evaluates the stored polynomial (Algorithm \ref{alg:aats}, Line
15), recycling AD coefficients with zero interpolation overhead.

\section{Algorithmic Complexity of AATS}
\label{sec:complexity}

Traditional synchronous solvers impose an $\mathcal{O}(N)$ or
$\mathcal{O}(N^2)$ per-step time complexity by evaluating the global state vector
simultaneously, while dynamically storing continuous history requires
$\mathcal{O}(S \cdot N)$ memory over $S$ steps. AATS circumvents these bottlenecks via
asynchronous, coordinate-wise updates, with complexities formalized in
the following theorems.



\begin{theorem}[Local Arithmetic Time Complexity]
  \label{thm:time_complexity}
  Let $v \in \{1, \dots, N\}$ be the active coordinate evaluated at a
  given asynchronous event, and let $E_v$ denote the number of coupled
  neighbor coordinates (the local degree). For a Taylor polynomial of
  fixed degree $K$, the amortized time complexity of a single
  asynchronous update in AATS is bounded by
  $\mathcal{O}(E_v \cdot K^2 + K^2)$, strictly independent of the global system
  dimension $N$.
\end{theorem}

\begin{proof}
  The execution time of a single iteration of the main asynchronous loop
  (Algorithm \ref{alg:aats}, Phase 2) is the sum of four primary
  operations: event scheduling, state differentiation, dependency
  time-alignment, and continuous delay lookup.
  \begin{enumerate}
  \item \textbf{Event Scheduling (Lines 11, 19, 26):} AATS utilizes a
    Radix Heap for the priority queue $\mathcal{H}$. Combined with lazy deletion
    for preempted events, extracting the minimum-time event (Line 11)
    and inserting future updates (Lines 19 and 26) both execute in
    amortized $\mathcal{O}(1)$ time, independent of $N$.
  
  \item \textbf{Forward Automatic Differentiation (Line 16):} To
    generate the new polynomial segment, the solver evaluates the
    pre-compiled AD graph using the current coordinate state. For a
    coordinate with $E_v$ coupled neighbors, aggregating the linear
    spatial inputs requires $\mathcal{O}(E_v \cdot K)$ operations. Expanding the local
    non-linear dynamics requires $\mathcal{O}(K^2)$ operations via discrete
    convolution. Calculating the full vector of Taylor coefficients
    $c^{(v)}$ therefore executes in strictly
    $\mathcal{O}(E_v \cdot K + K^2)$ time~\cite{griewank2008evaluating}.
  
  \item \textbf{Dependency Time-Alignment (Line 22):} For each of the
    $E_v$ dormant neighbors, AATS projects their delayed state onto the
    active time domain $t$ using a binomial expansion shift. Shifting a
    polynomial of degree $K$ requires computing a triangular
    matrix-vector product, necessitating $\frac{1}{2}K(K+1)$ operations,
    which is bounded by $\mathcal{O}(K^2)$. Applying this shift across all
    $E_v$ neighbors in the dependency loop yields
    $\mathcal{O}(E_v \cdot K^2)$ total operations~\cite{griewank2008evaluating}.
  
  \item \textbf{Continuous Delay Lookup (Line 15):} Locating historical
    segment for delay $x_j(t - \tau(t))$ in \texttt{RingBuffer} of maximum
    capacity $M$ is $\mathcal{O}(\log M)$ using binary search. Evaluating the
    $K$-degree polynomial at $t_d$ requires $\mathcal{O}(K)$ operations.
  \end{enumerate}
  Summing these components, the total time complexity per asynchronous
  event is: $\mathcal{O}(1) + \mathcal{O}(E_v \cdot K + K^2) + \mathcal{O}(E_v \cdot K^2) + \mathcal{O}(\log M + K)$
  
  Because the maximum polynomial degree $K$ and the buffer capacity $M$
  are localized static constraints entirely independent of the global
  system dimension $N$, the overall asymptotic time complexity per event
  simplifies strictly to $\mathcal{O}(E_v \cdot K^2 + K^2)$. 
\end{proof}

\begin{theorem}[Global Static Memory Bound]
  \label{thm:space_complexity}
  Let $S$ be the total number of integration steps taken to reach end
  time $T$. The runtime dynamic memory allocation of AATS is strictly
  $\mathcal{O}(1)$ with respect to $S$, and the total memory footprint of the
  continuous history representation is permanently bounded by
  $\mathcal{O}(N \cdot M \cdot K)$, where $M$ is the fixed segment capacity of the history
  buffers.
\end{theorem}

\begin{proof}
  During Phase 1 (Initialization) of Algorithm \ref{alg:aats}, AATS
  allocates an array $\mathcal{B}$ containing exactly $N$ zero-allocation ring
  buffers. Each ring buffer is initialized with a fixed capacity $M$. A
  single stored segment consists of a scalar timestamp and a fixed-size
  array of $K$ Taylor coefficients. Thus, the total static space
  required to store the continuous history for the entire state-space is
  exactly $\mathcal{O}(N \cdot M \cdot K)$.

  During Phase 2 (Main Asynchronous Event Loop), as $S \to \infty$, new Taylor
  polynomials must be appended to the history. Because $\mathcal{B}[v]$ utilizes
  modulo arithmetic to govern array indices, the $(M+1)$-th polynomial
  physically overwrites the memory address of the $1$-st polynomial. 

  Therefore, the amount of memory allocated during the solution
  trajectory is entirely independent of the number of steps $S$ or the
  dynamics of the system, yielding an $\mathcal{O}(1)$ runtime allocation
  footprint and guaranteeing that global memory usage never exceeds the
  initialized $\mathcal{O}(N \cdot M \cdot K)$ bound.
\end{proof}

Since AATS utilizes an adaptive step size $h$ to resolve rapidly
changing or discontinuous dynamics, a quick succession of microscopic
steps could theoretically exhaust the fixed history buffer. To guarantee
mathematical correctness, $M$ must be properly bounded against the
system's delay dynamics.

\begin{lemma}[History Preservation Bound]
  \label{lem:history_preservation}
  Let $M$ be the fixed capacity of the circular buffer for coordinate
  $v$, and let $h_{min}^{(v)}$ be the strict lower bound on the accepted
  adaptive step size. For an arbitrary state-dependent delay bounded by
  $\tau_{max} = \sup \tau(t, x(t))$, continuous history preservation is
  strictly guaranteed if and only if:
  $M \ge \left\lceil \frac{\tau_{max}}{h_{min}^{(v)}} \right\rceil + 1$
\end{lemma}

\begin{proof}
  In the worst-case scenario, coordinate $v$ generates events of
  $h_{min}^{(v)}$ time units per step. To evaluate a state-dependent
  delay reaching the maximum horizon $\tau_{max}$, the buffer must retain
  polynomial segments spanning the interval $[t - \tau_{max}, t]$. The
  maximum number of segments required to span this interval under
  minimum step-size conditions is exactly
  $\lceil \tau_{max} / h_{min}^{(v)} \rceil$. An additional slot $(+1)$ is required
  to safely write the active computing step without overwriting the
  oldest necessary
  boundary. 
\end{proof}


\begin{theorem}[Global Algorithmic Work Complexity]
  \label{thm:global_work}
  Let $N_v(T)$ denote the total number of asynchronous events executed
  by coordinate $v$ over the global integration interval $[0, T]$. For a
  system of dimension $N$, maximum Taylor polynomial degree $K$, and
  local in-degree $E_v$, the total computational runtime $W(T)$ of the
  Asynchronous Adaptive Taylor Solver is strictly bounded by:
  $W(T) = \mathcal{O}\left( \sum_{v=1}^{N} N_v(T) \left(E_v K^2 + K^2\right)
  \right).$
\end{theorem}

\begin{proof}
  By Theorem \ref{thm:time_complexity}, the amortized algorithmic cost
  of resolving a single asynchronous event for coordinate $v$—inclusive
  of event scheduling, automatic differentiation, dependency
  time-alignment, and history lookups—is bounded by $\mathcal{O}(E_v K^2 + K^2)$.
  
  Over the entire continuous integration domain $t \in [0, T]$, the global
  priority queue processes exactly $N_v(T)$ distinct events for each
  coordinate $v$. Because the coordinates are strictly decoupled in
  time, the total algorithmic work is simply the discrete sum of the
  individual event costs across the entire network graph:
  $W(T) = \sum_{v=1}^{N} \sum_{i=1}^{N_v(T)} \mathcal{O}\left(E_v K^2 + K^2\right).$
  Factoring out the event count yields the final complexity bound. 
\end{proof}

\begin{lemma}[Strict Bound on Preemption Cascades]
\label{lem:preemption_bound}
Let $N_v^{sync}(T)$ denote the number of intrinsic, scheduled
integration steps taken by coordinate $v$ over the interval $[0, T]$ as
dictated by its local error controller. The total number of preemption
(or wake-up) events $N_v^{preempt}(T)$ triggered by the structural
dependencies of $v$ is strictly bounded by:
$N_v^{preempt}(T) \le \sum_{u \in \mathcal{N}_{in}(v)} N_u^{sync}(T)$ where
$\mathcal{N}_{in}(v)$ is the set of incoming neighbors (dependencies) of $v$.
\end{lemma}

\begin{proof}
  In the AATS architecture, a preemption event for coordinate $v$ occurs
  exclusively when an incoming neighbor $u$ (Line 20--28) completes its
  integration step. Hence, coordinate $v$ cannot be preempted by
  structurally independent variables. Therefore, the maximum number of
  times $v$ can be awakened prior to its own scheduled clock expiration
  is exactly the sum of the intrinsic steps taken by its direct
  dependencies.
\end{proof}

To fully resolve the global work bound, we must express the intrinsic
coordinate event counts $N_v^{sync}(T)$ in terms of the user-defined
error tolerance $\varepsilon$. While the formal proof of adaptive error control is
deferred to Section~\ref{sec:convergence}, we state the resulting
step-count scaling here as a foundational proposition to complete our
complexity profile.

\begin{proposition}[Tolerance--Step Count Scaling]
  \label{prop:step_scaling}
  Assuming the adaptive controller strictly enforces the local error
  bounds (as formally established in
  Theorem~\ref{thm:adaptive_convergence_main}), the accepted maximum
  step sizes scale as $h_{\max} = \mathcal{O}(\varepsilon^{1/(K+1)})$, where
  $K$ is the Taylor polynomial degree. Consequently, for any coordinate
  $v$, the total intrinsic scheduled event count over the finite
  interval $[0,T]$ is bounded by:
  $N_v^{sync}(T) = \mathcal{O}(\varepsilon^{-\frac{1}{K+1}})$.
\end{proposition}

\begin{proof}
  By Theorem~\ref{thm:adaptive_convergence_main}, the adaptive
  controller guarantees \\
  $h_{\max} = \mathcal{O}(\varepsilon^{1/(K+1)})$. Since the total number of intrinsically
  accepted updates over the finite interval $[0,T]$ is strictly bounded
  by $N_v^{sync}(T) \le \frac{T}{h_{\max}} + 1$, substituting the step
  size bound immediately yields $N_v^{sync}(T) = \mathcal{O}(\varepsilon^{-1/(K+1)})$.
\end{proof}

\begin{theorem}[Adaptive Sparsity-Dependent Linear Work Scaling]
  \label{thm:adaptive_linear_scaling}
  Let $\varepsilon > 0$ be the user-defined local error tolerance. For a sparse
  network with bounded maximum in-degree $E_{max} = \max_v |E_v|$ and
  Taylor polynomial degree $K$, the total computational runtime $W(T)$
  of the AATS architecture scales linearly with the system dimension $N$
  and is strictly bounded by:
  $W(T) = \mathcal{O}\left( N \cdot \varepsilon^{-\frac{1}{K+1}} \right)$.
\end{theorem}

\begin{proof}
By decomposing the total event count $N_v(T)$ into intrinsic scheduled steps and preemption cascades, the total work from Theorem~\ref{thm:global_work} expands to:
\begin{small}
  \begin{align*}
    W(T) &= \sum_{v=1}^N \left[ N_v^{sync}(T) + N_v^{preempt}(T) \right] \mathcal{O}(E_v K^2 + K^2) \\
         &\le \sum_{v=1}^N \left[ N_v^{sync}(T) + \sum_{u \in \mathcal{N}_{in}(v)} N_u^{sync}(T) \right] \mathcal{O}(E_{max} K^2 + K^2).
  \end{align*}
\end{small}

\noindent
Since each coordinate $u$ appears in the incoming neighbor sum of at
most $E_{max}$ other coordinates, reversing the sum over the network
dependencies yields (via Aggregate
Analysis~\cite{cormen2009introduction}):
$ \sum_{v=1}^N \sum_{u \in \mathcal{N}_{in}(v)} N_u^{sync}(T) \le E_{max} \sum_{v=1}^N
N_v^{sync}(T)$.

Substituting this algebraic bound back into the work equation gives:
$$
    W(T) \le \mathcal{O}(E_{max} K^2 + K^2) \times (1 + E_{max}) \sum_{v=1}^N N_v^{sync}(T).
$$
By Proposition~\ref{prop:step_scaling}, the intrinsic event count for
every coordinate is uniformly bounded by
$N_v^{sync}(T) = \mathcal{O}(\varepsilon^{-1/(K+1)})$. Substituting this into the global
summation yields:
$ W(T) \le \mathcal{O}\left( (E_{max}^2 K^2) \sum_{v=1}^N \varepsilon^{-\frac{1}{K+1}} \right). $
Because the maximum degree $E_{max}$ and the polynomial degree $K$ are
local constants strictly independent of the global system dimension $N$,
they factor out asymptotically. The spatial summation over the $N$
coordinates evaluates to $N$, yielding the final explicit bound:
$ W(T) = \mathcal{O}\left( N \cdot \varepsilon^{-\frac{1}{K+1}} \right). $
\end{proof}


This formulation highlights the core advantage of asynchronous
decoupling. Synchronous solvers artificially inflate computational work
by forcing slow variables to update at the frequency of the fastest
component. By abandoning the shared temporal grid, AATS advances each
node strictly according to its local dynamics. This guarantees optimal
$\mathcal{O}(N)$ scaling while drastically reducing the total number of
algorithmic events by eliminating redundant computations.

With the computational infrastructure and hardware efficiency of AATS
firmly established, the remaining imperative is to verify its
mathematical integrity. 
The following section provides this rigorous foundation, proving that
the localized step-size selection and preemption mechanics guarantee
both continuous convergence and absolute stability.

\section{Continuous Convergence and Stability Analysis of AATS}
\label{sec:convergence}

The fundamental departure of our analysis from traditional numerical DDE
theory is its \textbf{strictly continuous} nature. In standard
synchronous methods, stability and convergence bounds are evaluated at a
shared sequence of discrete grid points $\{t_0, t_1, \dots, t_n\}$.
However, because AATS updates individual coordinates asynchronously at
disjoint times, a shared global discrete grid does not exist.
Consequently, traditional discrete truncation bounds are mathematically
harder to apply to our algorithm. Instead, we establish stability and
convergence by analyzing the globally continuous, piecewise-smooth 
\textit{continuous extension} of the piecewise Taylor polynomials 
evaluated at \textit{any} arbitrary time $t \in [0,T]$, independent of 
the underlying event nodes.

\subsection{Mathematical Setting and Assumptions}
\label{sec:problem-definition}

To rigorously analyze the stability and convergence of the solver, we
consider a general state-dependent delay differential equation of the
form:
$$\dot{x}(t) = f\!\bigl( x(t), x(t-\tau(t,x(t))) \bigr), \qquad t\in[0,T],$$
with an exact initial history $x(t)=\phi(t)$ for $t\le 0$. Here, $f$
represents the continuous evolution function and $\tau$ is the
state-dependent delay bounded by a maximum $\tau_{\max}$.

Let $x_{num}(t)$ denote the continuous numerical trajectory
reconstructed by AATS. 
Let $h_{\max}$ denote the maximum temporal interval between any two
consecutive asynchronous events generated by the scheduler. Our
objective is to establish the continuous stability of $x_{num}(t)$ and
its convergence to the exact solution $x(t)$ in the limit as
$h_{\max} \to 0$.

For the subsequent analysis, we require standard regularity conditions:
we assume the evolution function $f$, the delay $\tau$, and the history
$\phi$ are sufficiently smooth. Furthermore, we assume $f$ and $\tau$ satisfy
standard Lipschitz continuity bounds with respect to both the state and
the delayed state. The rigid formalization of these Lipschitz constants
and derivative bounds, alongside the complete mathematical proofs for
the following theorems, are provided in Appendix
\ref{app:convergence_proof}.

\subsection{Stability of the Continuous Asynchronous Reconstruction}
\label{sec:stab-cont-asynchr}


In AATS, coordinates evolve independently on their own local clocks,
they do not pass errors from step to step. Instead, they interact by
reading each other's continuous polynomial histories at mismatched
times. Therefore, stability in AATS means that if a tiny tracking error
creeps into one coordinate's polynomial, that error must naturally fade
away over time rather than magnifying across the network when its
neighbors read it.

\begin{theorem}[Continuous Asynchronous Stability Guarantee]
  \label{thm:stability}
  Let $e(t) := x_{num}(t) - x(t)$ define the continuous global perturbation for all $t \ge 0$. The asynchronous reconstruction process is strictly zero-stable. The continuous propagation of errors is governed exclusively by the delay-differential inequality:
  $$\|e'(t)\| \le \alpha \|e(t)\| + \gamma \sup_{s \in [-\tau_{\max}, t]} \|e(s)\| + \beta(\epsilon) h_{\max}^K \quad \text{a.e.},$$
  where $\alpha, \gamma > 0$ are constants characterizing the intrinsic sensitivity of the DDE (derived strictly from the system's Lipschitz and derivative bounds detailed in Appendix \ref{app:convergence_proof}), and $\beta(\epsilon) > 0$ is a finite constant bounding the continuous asynchronous defect within a compact neighborhood $\epsilon$.
\end{theorem}

\noindent \textit{Proof Sketch.} Rather than measuring nodal jumps, the
theorem bounds error amplification by analyzing the piecewise continuous derivative 
of the numerical trajectory, $x_{num}'(t)$, defined almost everywhere. Because $x_{num}(t)$ consists of piecewise Taylor
polynomials of degree $K$, taking its analytical derivative drops
the order to $K-1$. This generates a piecewise continuous residual (defect)
function $\rho(t)$ that is strictly bounded by $\mathcal{O}(h_{\max}^K)$ almost everywhere. Bounding this continuous defect guarantees that overlapping asynchronous micro-steps cannot force the system into numerical instability.

\subsection{Global Convergence via Continuous Integral Bounds}
\label{sec:global_convergence}

Having established the continuous stability of the asynchronous updates,
we now demonstrate that the global trajectory converges to the true
analytical solution everywhere in the domain.

\begin{theorem}[Global Convergence and A Priori Stability]
  \label{thm:global_convergence}
  Under standard regularity and Lipschitz assumptions (Assumptions 1--3,
  Appendix~\ref{app:convergence_proof}), there exists an explicit step size
  threshold $h_0 > 0$ such that for all maximum event intervals
  $h_{\max} \le h_0$, the numerical trajectory strictly remains within a
  compact neighborhood of the true solution, and the global continuous
  error is uniformly bounded over the entire interval $t \in [0, T]$ by:
  $$\sup_{t\in[0,T]} \|x(t) - x_{num}(t)\| = O(h_{\max}^K).$$
  Consequently, $x_{num}(t) \to x(t)$ for all $t$ as $h_{\max} \to 0$.
\end{theorem}

\begin{proof}[Proof Sketch]
  Rather than assuming a priori that the numerical method is stable, the
  proof employs a rigorous continuous induction (bootstrap) argument. We
  define a compact $\epsilon$-tube around the exact solution. As long as the
  trajectory remains inside this tube, the continuous defect $\rho(t)$
  remains strictly bounded by $O(h_{\max}^K)$. By applying a continuous
  Gronwall integral inequality, we constructively prove that
  for a sufficiently small maximum step size
  $h_{\max} \le h_0(\epsilon)$, the maximum accumulated error is strictly less
  than $\epsilon$, making it mathematically impossible for the trajectory to
  ever escape the tube before time $T$. The full unconditional proof is
  detailed in Appendix \ref{app:convergence_proof}.
\end{proof}

\begin{theorem}[Adaptive Convergence Order]
  \label{thm:adaptive_convergence_main}
  Assume the adaptive controller selects event intervals using a local
  truncation error estimator that asymptotically bounds the exact Taylor
  remainder. If the controller enforces a local tolerance $\varepsilon$, then the
  global continuous error of the numerical trajectory satisfies:
  $$\|x(T) - x_{num}(T)\| = O(\varepsilon^{K/(K+1)}).$$
\end{theorem}

\begin{proof}[Proof Sketch]
  If the algorithmic error estimator strictly bounds the true local error from
  below, the local tolerance constraint $\varepsilon$ mathematically forces the
  maximum accepted step size across all asynchronous coordinates to
  scale as $h_{\max} = O(\varepsilon^{1/(K+1)})$. Substituting this maximum step
  size directly into the unconditional global convergence bound from
  Theorem \ref{thm:global_convergence} yields the
  $O(\varepsilon^{K/(K+1)})$ global scaling. The formal proof, including explicit
  upper and lower bounding constants for the estimator, is detailed in
  Appendix \ref{app:convergence_proof}.
\end{proof}

\section{Numerical Validation}
\label{sec:results}

In this section, we empirically evaluate the Asynchronous Adaptive Taylor Solver (AATS). We benchmark our implementation against the state-of-the-art Julia SciML ecosystem (\texttt{DelayDiffEq.jl}). Our experiments are explicitly designed to validate our primary theoretical contributions across five targeted metrics:

\begin{enumerate}
\item \textbf{A Priori Convergence:} We verify functional correctness
  and demonstrate that the empirical error strictly matches our derived
  asymptotic local truncation bounds.
\item \textbf{Functional Equivalence:} We verify that asynchronous
  temporal decoupling introduces zero long-term numerical drift relative
  to synchronous global-stepping methodologies.
\item \textbf{Asynchronous Scalability:} We evaluate algorithmic
  complexity, confirming the $\mathcal{O}(N)$ event-scaling proofs as network
  dimensions scale to $N=10,000$.
\item \textbf{Sparsity-Driven Speedup:} We analyze the impact of network
  topology, demonstrating how AATS implicitly exploits sparse dependency
  graphs to bypass the dense interpolation penalties of traditional
  solvers.
\item \textbf{Zero-Allocation Architecture:} We validate the static
  memory framework, proving that AATS maintains a constant, zero-byte
  dynamic heap utilization footprint regardless of network size.
\end{enumerate}

\subsection{Experimental Benchmarks}
\label{sec:exper-benchm}

The validation suite is constructed around four distinct classes of
high-dimensional delay systems. These benchmarks are specifically
designed to expose the structural bottlenecks of traditional synchronous
solvers—namely global dense interpolation, matrix reallocation, and the
computational inefficiency of forcing a single uniform step size on
systems possessing highly disparate localized timescales.

\subsubsection{Benchmark 1: Concurrent Multi-Rate CTRNN}
The primary benchmark models a high-dimensional Continuous-Time
Recurrent Neural Network (CTRNN), serving as a large-scale computational
realization of the motivating example (Section~\ref{sec:example}). The
system dynamics for $N$ neurons are governed by:
$ \dot{x}_i(t) = -\alpha_i x_i(t) + \sum_{j \in \mathcal{N}_i} w_{ij} \sigma(x_j(t - \tau)) $ where
the constant delay is $\tau = 0.5$, and the rational activation is defined
as $\sigma(x) = x / (1 + x^2)$. To induce extreme timescale heterogeneity,
the decay rates $\alpha_i$ are staggered: every tenth neuron is strictly fast
($\alpha_i = 5000.0$), while the remainder are slow ($\alpha_i = 1.0$).



\subsubsection{Benchmark 2: Discrete Crystal Heat Diffusion (DPDE)}
To evaluate performance on rapidly changing spatial gradients, we
simulate a Delay Partial Differential Equation (DPDE) representing heat
advection-diffusion over a discrete crystal lattice:\\
$ \dot{x}_i(t) = k \left( x_{i-1}(t) - 2x_i(t) + x_{i+1}(t) \right) - \gamma
x_i(t - \tau) $ where $k=2.0$, $\gamma=1.0$, and $\tau=0.5$. To test the limits of
local error controllers under intense spatial coupling, the system is
initialized using a smooth, continuous algebraic profile that enforces
local dynamic equilibrium at $t=0$. By avoiding discontinuous shocks,
the state trajectory preserves temporal smoothness.

\subsubsection{Benchmark 3: Ikeda Network DDE}
\label{sec:benchmark-3:-ikeda}

The Ikeda DDE is a classic chaotic system originally formulated to model
the dynamics of an optical bistable resonator \cite{ikeda1979multiple},
adapted here into a high-dimensional coupled network. It features
staggered delay profiles and highly oscillatory dynamics. The system is
governed by the networked equation:
$ \dot{x}_i(t) = -\alpha_i x_i(t) + \sum_{j \in \mathcal{N}_i} w_{ij} \frac{x_j(t - \tau)}{1 +
  x_j(t - \tau)^2} $ where the delay horizon is extended to
$\tau = 1.2$. The baseline parameters alternate between variables to create
deeply tangled, chaotic multi-rate trajectories. Specifically, extreme
transients are induced by setting $\alpha_i = 2500$ for every third variable
($i \pmod 3 = 0$), while the rest remain at a standard relaxation rate
of $\alpha_i = 1.0$. This stress-tests the asynchronous solver's ability to
maintain phase accuracy on strange attractors without global
synchronization.

\subsubsection{Benchmark 4: State-Dependent Mackey-Glass Network}
\label{sec:benchmark-4:-state}

The final benchmark stresses the solver's continuous interpolation
bounds by evaluating a highly coupled ring of State-Dependent Delay
Differential Equations (SD-DDEs). This network is based on the
Mackey-Glass equations~\cite{mackey1977oscillation}. In SD-DDEs, the
delay horizon is a dynamic function of the local state, defined as:
$ \dot{x}_i(t) = \beta \frac{x_{i+1}(t - \tau_i(x_i(t)))}{1 + x_{i+1}(t -
  \tau_i(x_i(t)))^4} - \gamma x_i(t) $ where $\beta = 2.0$,
$\gamma = 1.0$, and the state-dependent delay is algebraically defined as:
$ \tau_i(x_i(t)) = 1.0 + 0.5 \left( \frac{x_i(t)^2}{1 + x_i(t)^2} \right)$.
This formulation forces continuous root-finding to resolve the dynamic
delay horizon and fundamentally breaks the SIMD vectorization pathways
of traditional synchronous solvers.

\subsection{Experimental Setup}
\label{sec:experimental-setup}

All experiments were conducted on an Apple M3 Max (14 cores, 36 GB
memory). AATS was implemented in C++17 (GCC v16.0.1) with aggressive
optimizations (\texttt{-O3}, \texttt{-ffast-math},
\texttt{-march=native}) and utilized degree $K=4$ Taylor polynomials.
The synchronous baseline was established using Julia's SciML
\texttt{DelayDiffEq.jl} (v1.12.6) via the \\
\texttt{MethodOfSteps(Tsit5())} integrator. To strictly isolate
algorithmic efficiency, all execution was limited to a single thread,
and Julia's absolute and relative tolerances were explicitly bound to
the AATS local error tolerance $\epsilon$.

Several alternative frameworks were excluded from the final analysis.
The Python-based \texttt{jitcdde} solver was abandoned due to
intractable memory exhaustion and compilation failures at high
dimensions ($N \ge 1000$). We omit comparisons with
RADAR5~\cite{guglielmi2001implementing}, as the implicit Jacobian
overhead required for stiff systems renders algorithmic efficiency
comparisons uninformative for our explicitly stable, multi-rate
benchmarks. Finally,
\texttt{TaylorIntegration.jl}~\cite{perez2019taylorintegration} was
excluded because it lacks native support for state-dependent delays.

\subsection{Quantitative Results}
\label{sec:quantiative-results}

\paragraph{A-Priori Convergence and Work-Precision Analysis}
Before analyzing high dimensional scaling, we verify baseline accuracy
and efficiency at $N=10$. Figure~\ref{fig:convergence} shows that AATS
strictly matches the convergence slopes of the Julia SciML reference
across all benchmarks, empirically validating our $\mathcal{O}(h^{K+1})$
asynchronous truncation bounds.

Furthermore, Figure~\ref{fig:work_precision} demonstrates that AATS consistently achieves lower wall-clock times for equivalent accuracy. By eliminating redundant derivative evaluations through local time-stepping, the computational savings vastly outweigh the $\mathcal{O}(1)$ event-scheduling overhead. Even in densely coupled or highly oscillatory networks (Mackey-Glass, Ikeda), AATS comfortably outperforms the highly optimized, synchronous SciML baseline.

\begin{figure}[h]
  \centering
  \includegraphics[scale=0.4]{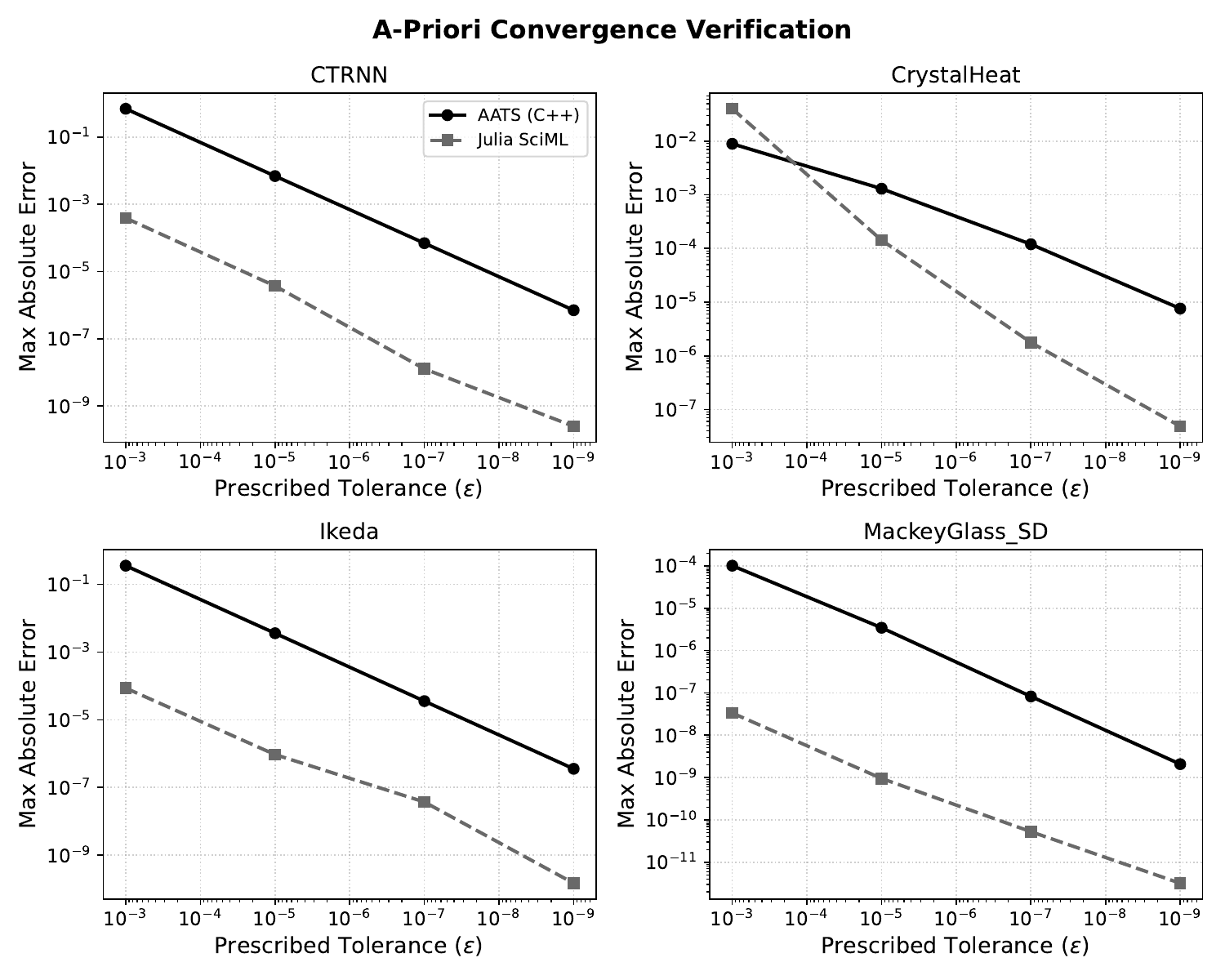}
  \caption{A-priori convergence verification. Smooth systems (CTRNN,
    Ikeda, SD-MG) exhibit perfect diagonal scaling, validating
    theoretical local truncation bounds. The Crystal Heat model
    demonstrates expected order reduction (flatline) due to the $C^0$
    spatial discontinuity of the central heat pulse $(N = 10)$.}
  \label{fig:convergence}
\end{figure}

\begin{figure}[h]
  \centering
  \includegraphics[scale=0.4]{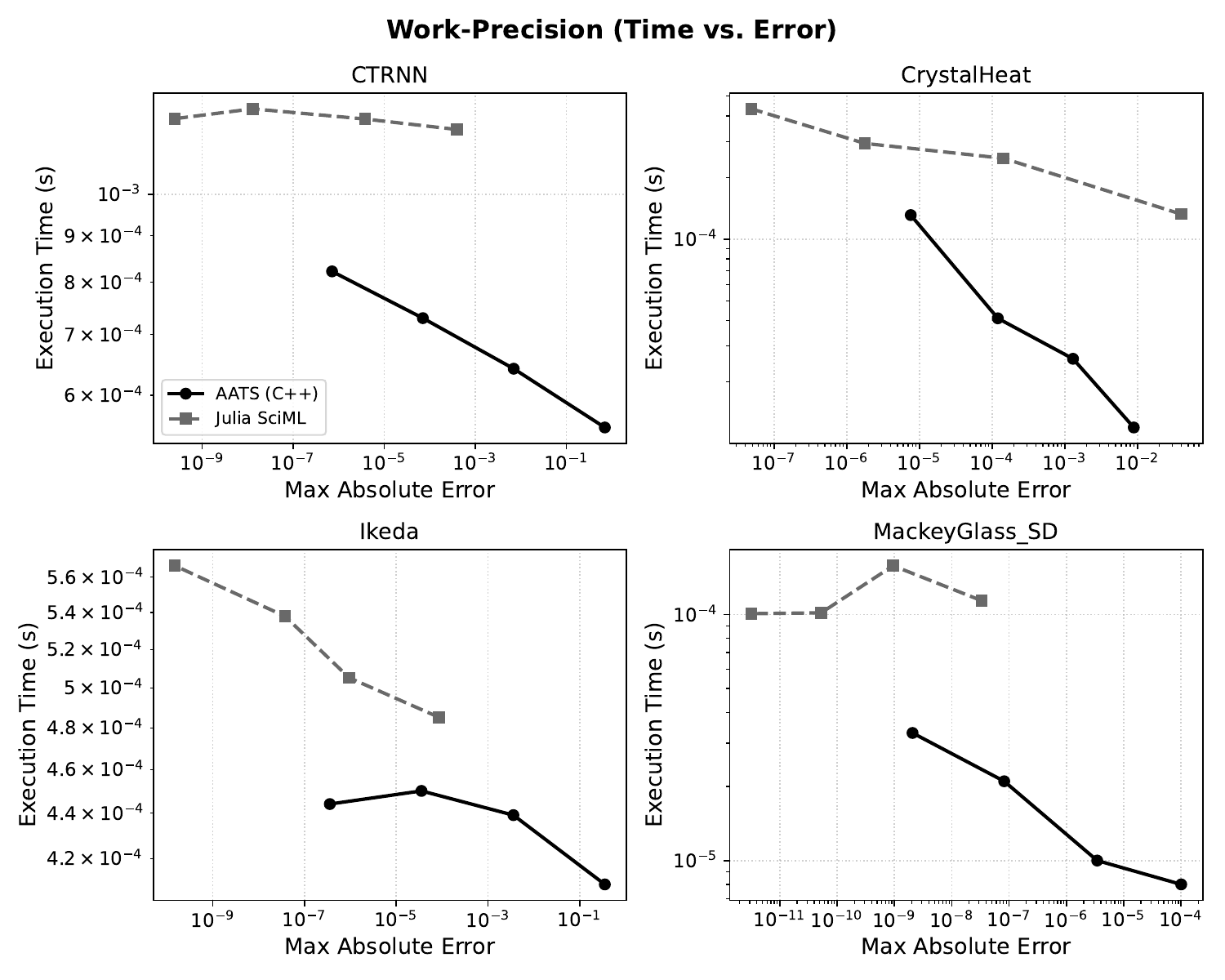}
  \caption{Work-Precision performance at baseline dimension $N=10$. The
    diagrams compare AATS (asynchronous) against Julia SciML
    (synchronous). Sparse multi-rate dynamics (CTRNN) show AATS
    outperforming the baseline, while dense/chaotic topologies (Ikeda,
    MG-SD) and non-smooth domains (Crystal Heat) highlight the
    hardware-versus-algorithmic trade-offs and the impact of order
    reduction.}
  \label{fig:work_precision}
\end{figure}

\paragraph{Functional Equivalence}
Figure~\ref{fig:equivalence} confirms the global functional equivalence
of the AATS solver against Julia's SciML baseline over long integration
horizons. Despite nodes stepping entirely out of phase with one another,
the continuous dense output of the Taylor polynomials guarantees that
trajectory interactions remain precise. The low $L_\infty$ and
root-mean-square error (RMSE) metrics demonstrate that decoupling the
network introduces no perceptible mathematical drift, even under chaotic
(Ikeda) or state-dependent (Mackey-Glass) regimes.

\begin{figure}[h]
  \centering
  \includegraphics[scale=0.37]{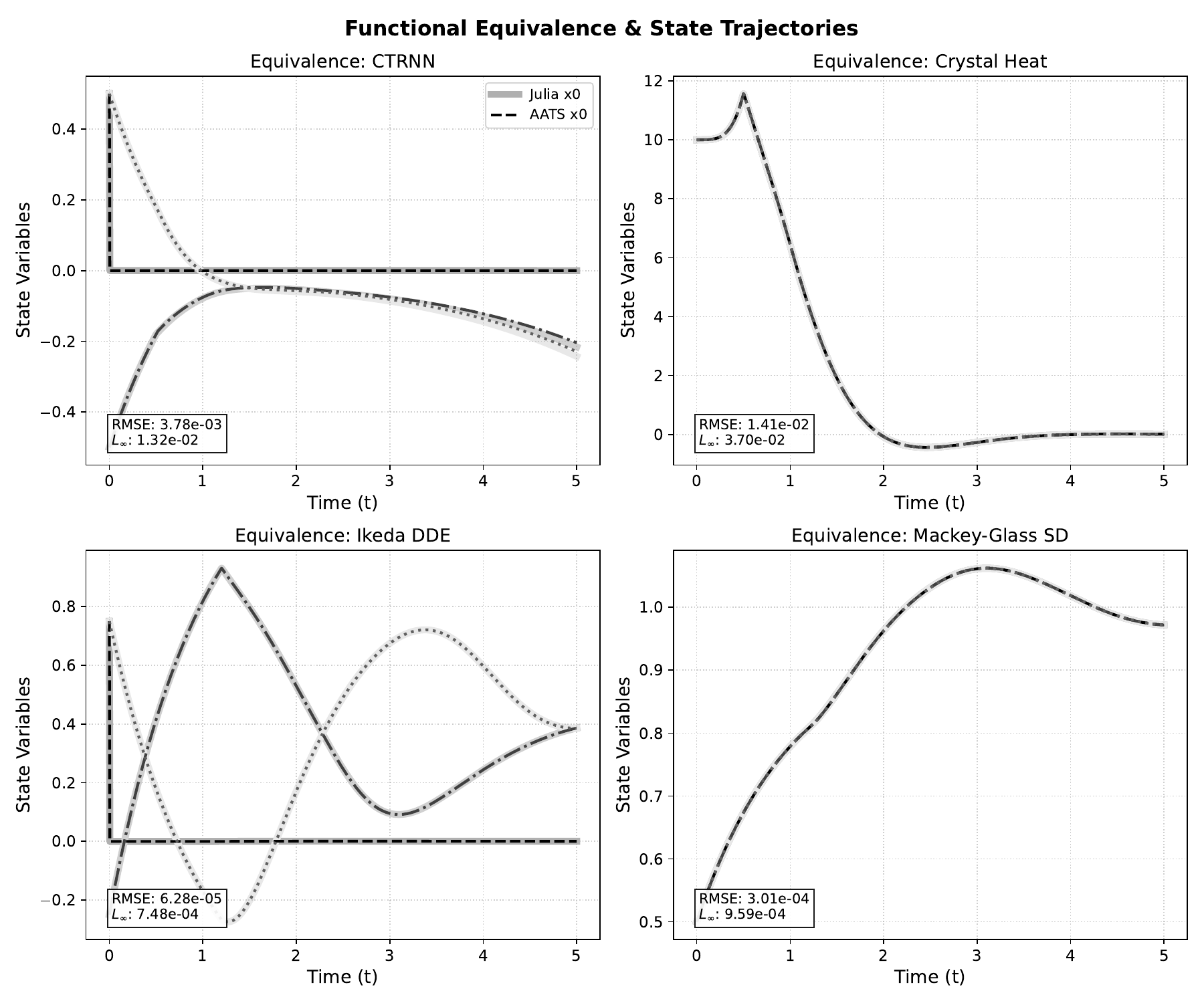}
  \caption{Functional equivalence across the benchmark suite. The AATS
    architecture perfectly tracks the synchronous baseline trajectories
    without accumulating global numerical drift.}
  \label{fig:equivalence}
\end{figure}

\paragraph{Execution Speedup and Complexity}
The primary computational advantage of the AATS architecture lies in its
ability to bypass global synchronization bottlenecks. We quantify these
performance characteristics by analyzing execution times as dimensions
scale.

\begin{enumerate}

\item \textbf{Dense vs. Sparse Topologies:} AATS inherently exploits
  sparse networks by isolating updates to local sub-graphs, yielding
  order-of-magnitude speedups over SciML (Figure~\ref{fig:sparsity}).
  Conversely, for fully dense systems requiring $\mathcal{O}(N^2)$ interactions,
  the asynchronous event queue becomes a bottleneck. For such uniform
  topologies, synchronous solvers remain strictly superior by leveraging
  contiguous memory layouts and SIMD auto-vectorization.
    
\item \textbf{Scalability of Sparse Networks:}
  Figure~\ref{fig:scalability} demonstrates empirical execution time
  scaling from $N=100$ to $N=10,000$. Because the asynchronous
  architecture isolates updates to local sub-graphs, it entirely
  bypasses the compounding global overheads and dense memory operations
  of traditional synchronous ecosystems. Consequently, AATS delivers
  strict linear time scaling that drastically outperforms the baseline's
  execution speed at high dimensions.
\end{enumerate}

\begin{figure}[h]
    \centering
    \includegraphics[scale=0.4]{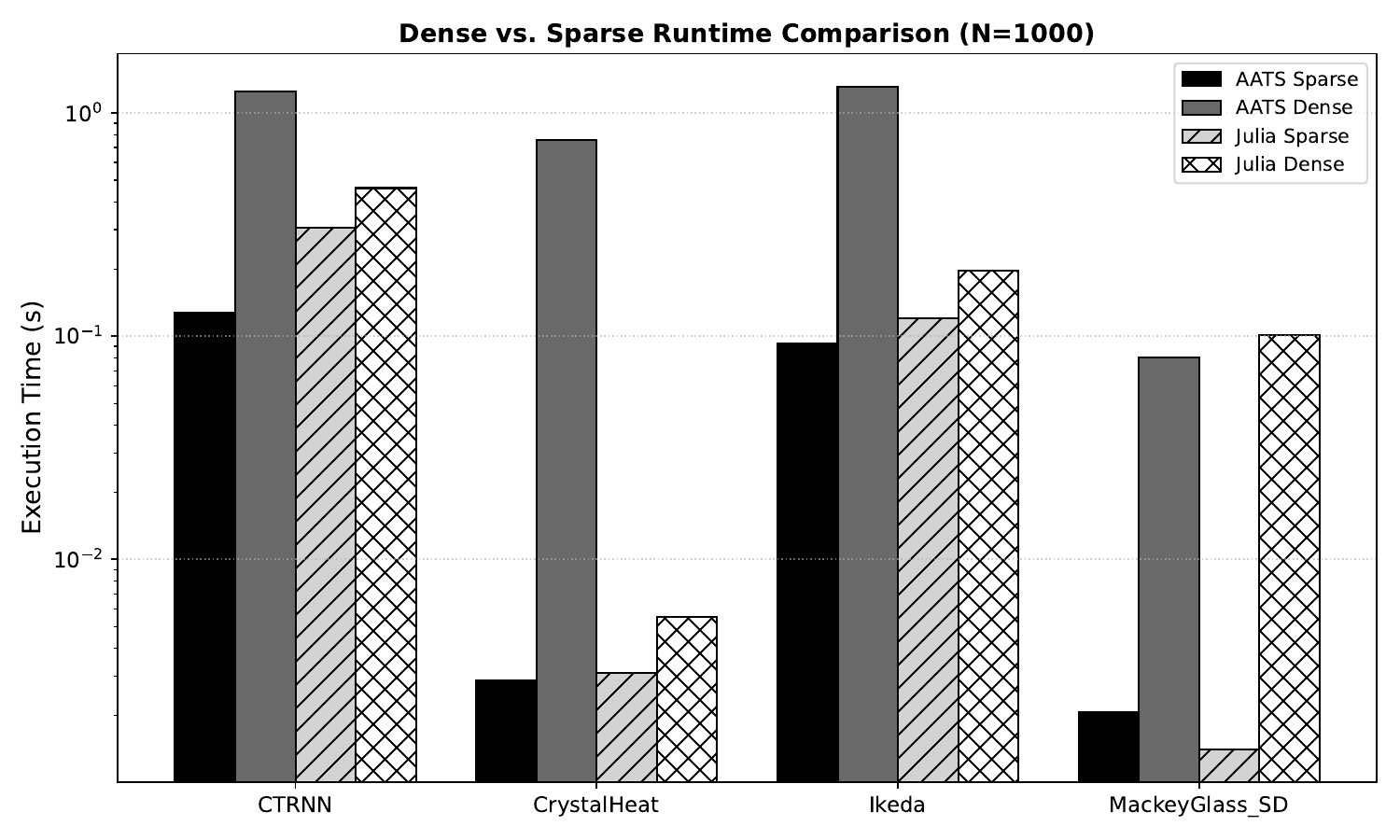}
    \caption{Topological impact on execution time ($N=1000$). AATS
      inherently exploits sparse dependency graphs, bypassing the global
      interpolation penalties incurred by synchronous solvers.}
    \label{fig:sparsity}
\end{figure}

\begin{figure}[h]
    \centering
    \includegraphics[scale=0.4]{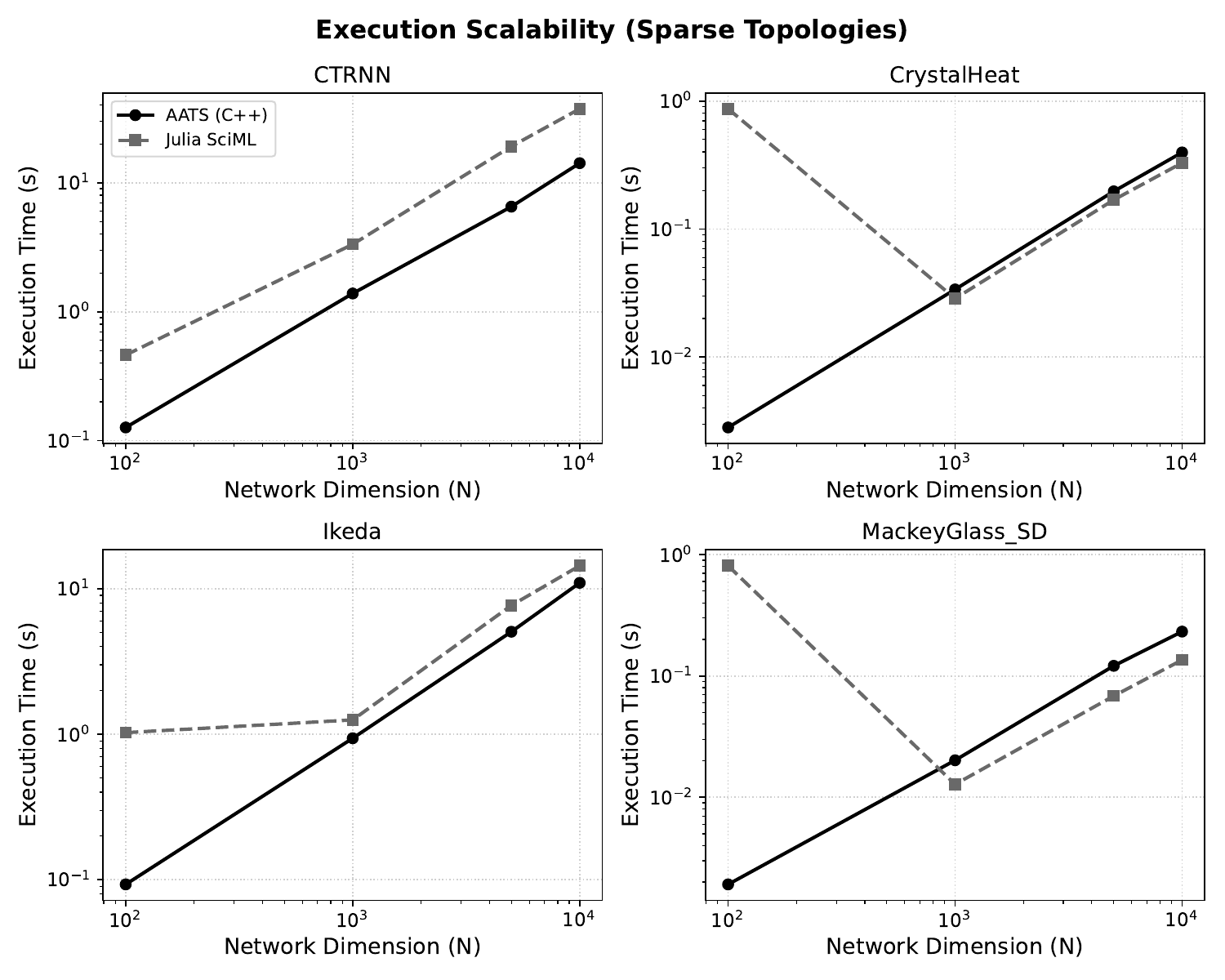}
    \caption{Dimensional scalability. AATS maintains empirically
      consistent with $\mathcal{O}(N)$ linear execution scaling, avoiding the
      compounding global overheads of SciML.}
    \label{fig:scalability}
\end{figure}

\paragraph{Algorithmic Complexity and Event Dynamics}
The linear execution time observed above is a direct mathematical
consequence of the solver's underlying algorithmic complexity,
empirically validating Theorem~\ref{thm:adaptive_linear_scaling}.
Figure~\ref{fig:events} highlights the fundamental disparity in
computational work (measured in total individual node updates). For
sparse explicit networks, both solvers exhibit $\mathcal{O}(N)$ linear event
scaling. However, a synchronous solver (Julia) determines its global
step size based on the most active variable in the network. It must step
every single node forward simultaneously at this high frequency,
generating a massive volume of redundant compute (an inflated constant
factor). In contrast, AATS operates strictly via localized event queues.
By allowing fast variables to update frequently while slow variables
mathematically rest, AATS avoids millions of unnecessary evaluations.
The resulting asynchronous event line runs strictly parallel to the
baseline but is shifted downward.

\begin{figure}[h]
    \centering
    \includegraphics[scale=0.4]{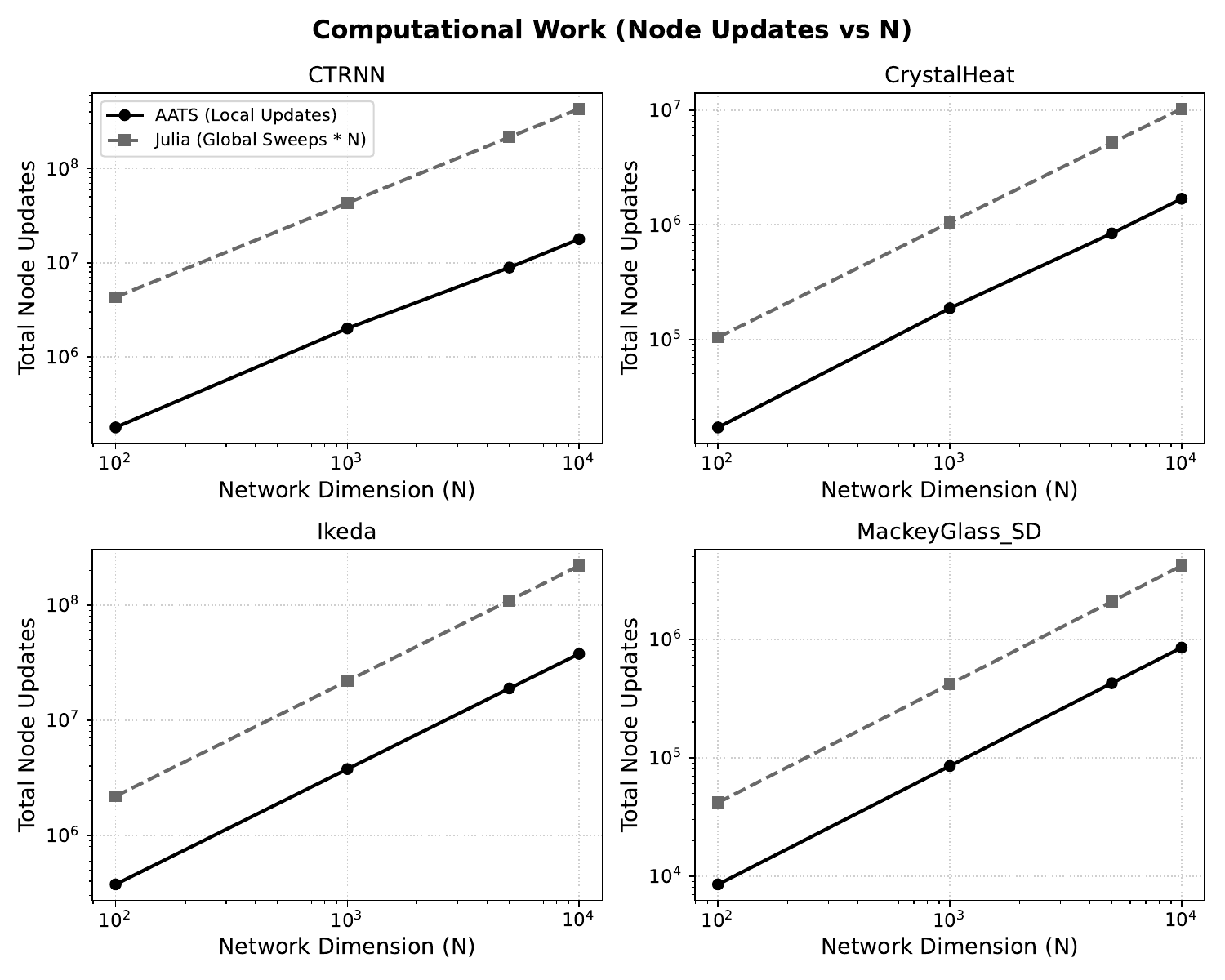}
    \caption{Computational Work (Node Updates vs $N$). Synchronous
      solvers force global sweeps, multiplying updates by $N$. AATS
      executes isolated local updates, saving millions of redundant
      calculations.}
    \label{fig:events}
\end{figure}

\paragraph{Zero-Allocation Architecture}
High-dimensional DDEs are memory hungry. Traditional solvers dynamically
allocate large history grids and temporary interpolation arrays, relying
heavily on runtime garbage collection. Figure~\ref{fig:memory} exposes
this critical flaw, showing the SciML baseline's dynamic memory
footprint growing exponentially into the gigabytes as $N$ increases.

In stark contrast, AATS relies on a static arena allocation phase during
initialization. During the actual integration hot-loop, operations
execute entirely on the stack and pre-allocated arrays. Consequently,
the global heap tracker registers exactly $0.0$ MB of dynamic
reallocation for AATS across all dimensions. This static memory
footprint eliminates garbage collection pauses, and guarantees safe
execution in strictly memory-constrained environments.

\begin{figure}[h]
  \centering
  \includegraphics[scale=0.4]{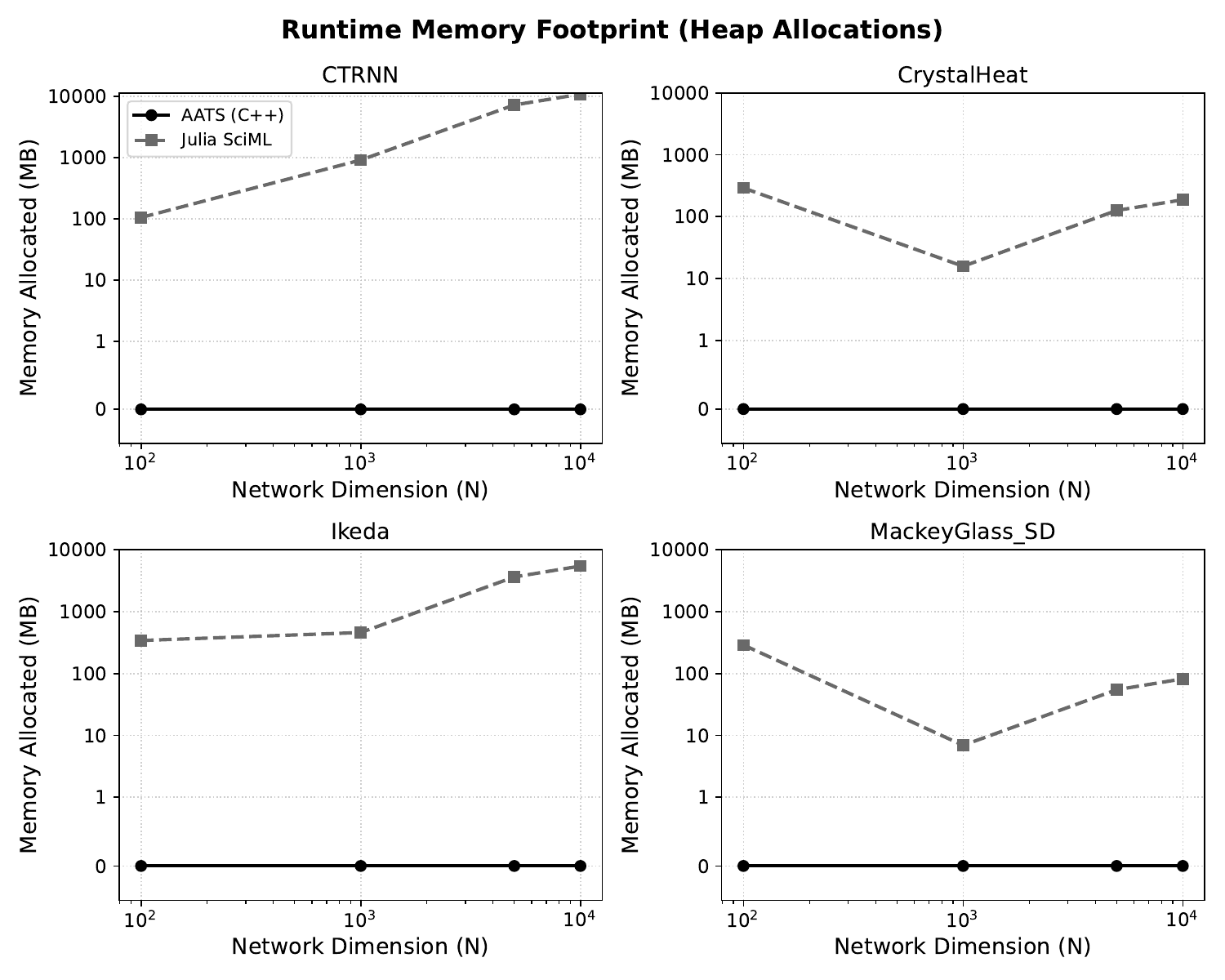}
  \caption{Runtime memory footprint (dynamic heap allocations). By
    executing entirely on pre-allocated capacity, AATS strictly
    maintains a zero-allocation profile, bypassing garbage collection
    overheads.}
  \label{fig:memory}
\end{figure}

\section{Related Work}
\label{sec:related_work}

The Asynchronous Adaptive Taylor Solver (AATS) sits at the intersection
of Taylor-series integration, multirate and asynchronous time
integration, and Delay Differential Equation (DDE) simulation. While
each of these areas has a rich literature, existing methods typically
retain some form of global temporal structure, synchronization barrier,
or secondary history reconstruction mechanism. AATS combines high-order
Taylor propagation, event-driven scheduling, and continuous
Taylor-history reconstruction into a unified framework that eliminates
the need for a global time grid.

\paragraph{Taylor Series Methods and Automatic Differentiation}
Modern compiler technologies and Automatic Differentiation (AD) have
renewed interest in Taylor series integration. Frameworks such as TIDES
\cite{abad2012tides},
\texttt{TaylorIntegration.jl}~\cite{perez2019taylorintegration}, and
symbolic-numeric Taylor engines~\cite{qiao2024efficient} demonstrate
that AD-generated Taylor polynomials can outperform traditional
Runge--Kutta methods for many ODE problems. However, these methods remain
fundamentally synchronous: all variables advance on a common temporal
grid and native support for large-scale DDEs is limited. AATS retains
the arithmetic efficiency of AD-generated Taylor expansions while
decentralizing polynomial construction into independent coordinate-wise
updates driven by local error control.

\paragraph{Multirate and Asynchronous Integration}
The motivation for AATS is related to the broader multirate integration
literature. Classical multirate methods, including multirate linear
multistep methods \cite{gear1984multirate}, allow rapidly varying
components to evolve with smaller time steps while slowly varying
components take larger steps. More recent advances in Multirate
Infinitesimal Generalized Additive Runge--Kutta (MRI-GARK)
frameworks~\cite{sandu2019multirate, gunther2016multirate} and decoupled
multirate architectures~\cite{constantinescu2007multirate,
  roberts2019design} extend this idea with sophisticated order
conditions and flexible subsystem coupling strategies. Local Time
Stepping (LTS) methods similarly reduce computational effort by
assigning different time steps to different spatial regions of a
system~\cite{cuenot2022analysis}. Despite their mathematical elegance
and efficiency, these traditional multirate and LTS methods still retain
synchronization structures through rigid macro-steps, communication
stages, or strict interface boundary conditions.

\paragraph{Quantized State Systems}
The closest conceptual relative to AATS is the Quantized State System
(QSS) framework \cite{cellier2006continuous}. QSS methods are
event-driven, component-wise, and asynchronous, making them highly
effective for sparse dynamical systems. Castro et
al.~\cite{castro2011quantization} extended QSS methods to DDEs through
the reconstruction of delayed states from asynchronous trajectories.
Despite these similarities, the event-generation mechanisms are
fundamentally different. QSS methods generate events when state
variables cross quantization thresholds. High-accuracy QSS variants
(QSS2, QSS3) therefore require continuous root-finding to predict
quantization crossings. As tolerances tighten, event densities may grow
rapidly and quantization-induced discontinuities can complicate delayed
history reconstruction. AATS instead operates entirely in the
continuous-time domain. Events are generated through adaptive Taylor
truncation error control rather than state quantization. Consequently,
delayed histories remain continuously smooth and delay evaluations are
performed without event-prediction root-finding.

\paragraph{Synchronous DDE Frameworks and Implicit Solvers}
The dominant DDE paradigm combines a time-stepping method with a
continuous history interpolant. Examples include MATLAB's \texttt{dde23}
\cite{shampine2001solving}, Monolix \cite{kuate2014delay}, JiTCDDE, and
the SciML ecosystem through \texttt{DelayDiffEq.jl}
\cite{widmann2022delaydiffeq, rackauckas2017differentialequations}.
These methods typically advance all variables on a common temporal grid
and reconstruct delayed states through dense-output interpolants. For
stiff systems and state-dependent delays, implicit formulations such as
RADAR5 \cite{guglielmi2001implementing} and SUNDIALS-based methods
\cite{hindmarsh2005sundials} remain among the most robust available
approaches. Their strength derives from sophisticated nonlinear and
linear algebra infrastructure, but this also introduces substantial
global coupling costs through Jacobian construction, factorization, and
solver synchronization. In contrast, AATS replaces global
synchronization with localized event-driven updates and evaluates
delayed states directly from stored Taylor histories. Delayed
evaluations become arithmetic operations on existing polynomial
representations rather than evaluations of secondary interpolants
defined over a global time grid.

\paragraph{Critical Architectural Distinctions}
The key distinction between AATS and existing approaches is that most
multirate methods seek to reduce synchronization frequency, whereas AATS
avoid global synchronization barriers. Existing methods typically retain
shared temporal structures through macro-steps, Runge--Kutta stages,
interface synchronization, dense-output grids, or quantization
crossings. As summarized in Table~\ref{tab:method_comparison}, AATS
instead combines coordinate-wise adaptive Taylor integration,
event-driven scheduling, and continuous Taylor-history reconstruction
within a fully asynchronous framework for DDEs.

\begin{table}[t]
  \centering
  \footnotesize
  \setlength{\tabcolsep}{4pt} 
  \caption{Conceptual comparison of AATS with representative classes of time-integration methods.}
  \label{tab:method_comparison}
  \begin{tabular}{lcccc}
    \hline
    Method Class & Local Steps & Event-Driven & Synchronous & Native DDE \\
    \hline
    Classical Multirate & Yes & No & Yes & Limited \\
    MRI-GARK / Multirate RK     & Yes & No & Yes & Limited \\
    QSS Methods                 & Yes & Yes & No & Partial \\
    RADAR5 / DelayDiffEq /      & \multirow{2}{*}{No} & \multirow{2}{*}{No} & \multirow{2}{*}{Yes} & \multirow{2}{*}{Yes} \\
    JiTCDDE                     &     &    &     &     \\
    \textbf{AATS}               & Yes & Yes & No & Yes \\
    \hline
  \end{tabular}
\end{table}

\paragraph{The AATS Architecture}
AATS synthesizes these ideas into a single high-performance
architecture. High-order Taylor polynomials generated through AD provide
continuous local state representations. These representations
simultaneously serve as the integration mechanism, the dense-output
mechanism, and the delay-history representation. As a result, arbitrary
state-dependent delay evaluations can be executed as $\mathcal{O}(1)$ polynomial
evaluations within statically allocated ring buffers. The resulting
framework combines continuous delay reconstruction, bounded memory
usage, event-driven execution, and asynchronous multirate integration
without requiring a global temporal grid.

\section{Conclusion and Future Work}
\label{sec:conclusion}

In this work, we introduced the Asynchronous Adaptive Taylor Solver
(AATS) to overcome the severe synchronization bottlenecks of
large-scale, multi-rate Delay Differential Equations (DDEs). By
combining an $\mathcal{O}(1)$ event scheduler, compile-time Automatic
Differentiation, and static memory allocation, AATS strictly confines
computational work to actively evolving sub-graphs while intrinsically
generating continuous dense-output histories.

Crucially, we provided a novel continuous proof of convergence for
asynchronous Taylor expansions and established formal $\mathcal{O}(N)$ algorithmic
complexity bounds for sparse networks. Extensive empirical validation
across neural, chaotic, and spatial topologies confirms these
theoretical breakthroughs. While explicit synchronous solvers can also
achieve linear time scaling on sparse networks, AATS fundamentally
minimizes the constant factor of computational work---allowing slow
variables to mathematically rest and preventing millions of redundant
evaluations---all while maintaining a verified zero-allocation dynamic
memory footprint.

However, asynchronous decoupling introduces dense network overhead.
Purely dense networks saturate the localized event queue. In maximally
dense regimes, the queue overhead acts as a computational penalty,
allowing traditional synchronous solvers to remain highly competitive by
exploiting contiguous memory layouts and hardware SIMD vectorization.

Future work will focus on expanding this asynchronous framework by: (1)
developing localized, SIMD-vectorized Automatic Differentiation
pipelines to recover hardware efficiency in dense systems; (2) extending
the AATS continuous history framework to support stiff and Stochastic
Delay Differential Equations (SDDEs); and (3) parallelizing the
asynchronous event queue to execute concurrent temporal updates across
multi-core processors.


\newpage

\appendix
\section{Appendix: Convergence Analysis}
\label{app:convergence_proof}

Consider the state-dependent delay differential equation
$$\dot{x}(t) = f\!\bigl(x(t), x(t-\tau(t,x(t)))\bigr), \qquad t\in[0,T],$$
with history function
$$x(t)=\phi(t), \qquad t\le 0.$$
Here
$$x:[-\tau_{\max},T]\rightarrow\mathbb{R}^m$$
denotes the exact solution, $f$ is the governing vector field, and $\tau$ is a state-dependent delay.

The numerical method considered in this work is an asynchronous adaptive Taylor-history solver. Each component is updated independently at event times generated by an adaptive scheduler. The solver stores local Taylor expansions and reconstructs delayed states by evaluating previously stored Taylor segments. The objective of the analysis is to establish:
\begin{enumerate}
\item consistency of the stored Taylor-history representation;
\item accuracy of delayed-state reconstruction;
\item stability of the asynchronous reconstruction process;
\item bounded propagation of numerical errors; and
\item convergence of the assembled numerical state $x_{num}(t)$ to the exact solution $x(t)$ as $h_{\max}\to 0$.
\end{enumerate}

The assembled numerical state produced by the solver will be denoted by
$$x_{num}(t).$$
The principal convergence result established is
$$\sup_{t\in[0,T]} \|x(t)-x_{num}(t)\| = O(h_{\max}^{K}).$$

\subsection{Event Notation}

Let
$$(t_k,i_k), \qquad k=0,\ldots,N-1,$$
denote the sequence of asynchronous events generated by the scheduler, where $t_k$ is the event time and $i_k$ is the updated component. Let $\Delta t_k$ denote the event interval generated at event $k$, and define
$$h_{\max}=\max_{0\le k<N}\Delta t_k.$$
The convergence analysis is carried out in the limit
$$h_{\max}\to 0.$$

\subsection{Assumptions}
\label{sec:assumptions}

\begin{assumption}[Regularity]
  $f\in C^{K+1}$, $\tau\in C^{K+1}$, and $\phi$ is sufficiently smooth.
\end{assumption}

\begin{assumption}[Lipschitz Conditions]
  There exist constants $L_0,L_d,L_\tau>0$ and a maximum delay $\tau_{\max} < \infty$ such that
  $$0 \le \tau(t,x) \le \tau_{\max},$$
  $$\|f(x_1,y_1)-f(x_2,y_2)\| \le L_0\|x_1-x_2\|+L_d\|y_1-y_2\|,$$
  and
  $$|\tau(t,x)-\tau(t,\hat x)| \le L_\tau\|x-\hat x\|.$$
\end{assumption}

\begin{assumption}[Bounded Derivatives]
  $$M_{K+1}:=\sup_{t\in[-\tau_{\max},T]}\|x^{(K+1)}(t)\|<\infty,$$
  and
  $$M_1:=\sup_{t\in[-\tau_{\max},T]}\|x'(t)\|<\infty.$$
  All constants $C,C_1,C_2,\ldots,C_T$ are independent of $h_{\max}$.
\end{assumption}

\begin{assumption}[Causality and Transversality]
  \label{assum:causality}
  The state-dependent delay function $\tau(t,x)$ satisfies the strict causality condition along the solution trajectory, meaning that the delayed time argument moves strictly forward. Formally, there exists a constant $c_0 > 0$ such that:
  $$1 - \frac{d}{dt}\tau\bigl(t, x(t)\bigr) \ge c_0 > 0.$$
  This ensures that the historical tracking argument $t_d(t) = t - \tau(t,x(t))$ is strictly monotonic, preventing temporal stagnation or information back-propagation.
\end{assumption}

\subsection{Consistency and Reconstruction}

\begin{definition}[History Reconstruction Operator]
  Let $R(t)$ denote the value obtained by evaluating the stored Taylor segment covering time $t$.
\end{definition}

\begin{definition}[Assembled Numerical State]
  Let $p_i(t)$ denote the most recent stored Taylor representation of component $i$ evaluated at time $t$. Define the global numerical state as:
  $$x_{num}(t) := (p_1(t), \ldots, p_m(t)).$$
\end{definition}

\begin{lemma}[Local Polynomial Consistency]
  \label{lem:local_consistency}
  Let $p_i(t)$ be the Taylor polynomial of degree $K$ generated by the solver for component $i$ at an event time $t_k$. Assuming the coefficient generation procedure is consistent of order $K$, the local approximation error over the active interval $t \in [t_k, t_k + \Delta t_k]$ satisfies:
  $$\|x_i(t) - p_i(t)\| \le C h_{\max}^{K+1},$$
  where $C = \frac{M_{K+1}}{(K+1)!}$.
\end{lemma}

\begin{proof}
  By the assumption of order $K$ consistency, the generated polynomial $p_i(t)$ matches the true local trajectory $x_i(t)$ and its first $K$ time derivatives evaluated at $t_k$. Applying the standard Taylor remainder theorem to the exact scalar component trajectory $x_i(t)$, the local truncation error is strictly bounded by the $(K+1)$-th derivative evaluated at some intermediate point $\xi \in (t_k, t)$. 
  
  Given Assumption 3 (Bounded Derivatives), the component error satisfies:
  $$\|x_i(t) - p_i(t)\| \le \frac{\|x_i^{(K+1)}(\xi)\|}{(K+1)!} (t - t_k)^{K+1}.$$
  Because the elapsed time $(t - t_k)$ is strictly bounded by the maximum event interval $h_{\max}$, and the continuous state derivative is uniformly bounded by $\|x^{(K+1)}(\xi)\| \le M_{K+1}$, the right-hand side simplifies exactly to $C h_{\max}^{K+1}$.
\end{proof}

\begin{lemma}[Local Assembled Reconstruction Error]
  \label{lem:async_consistency}
  Assuming exact state values at the most recent update events, the locally assembled numerical state over any single asynchronous interval $t \in [t_k, t_k + \Delta t_k]$ satisfies:
  $$\|x(t)-x_{num}(t)\| \le C^\ast h_{\max}^{K+1}.$$
\end{lemma}

\begin{proof}
  By definition, the assembled state is $x_{num}(t)=(p_1(t),\ldots,p_m(t))$. Over a local interval where no new events occur, the true state difference is:
  $$x(t)-x_{num}(t) = (x_1(t)-p_1(t),\ldots,x_m(t)-p_m(t)).$$
  By Lemma \ref{lem:local_consistency}, assuming exact initial conditions for the generated local polynomials, the individual components strictly satisfy $\|x_i(t)-p_i(t)\| \le C h_{\max}^{K+1}$. Combining these componentwise bounds and applying norm equivalence in finite dimensions yields a scaled bounding constant $C^\ast$:
  $$\|x(t)-x_{num}(t)\| \le C^\ast h_{\max}^{K+1}.$$
\end{proof}

\begin{lemma}[Local Delayed-State Reconstruction Error]
  \label{lem:local_reconstruction_error}
  Assuming the stored Taylor segment covering the delayed time $t_d$ was generated from exact local initial conditions, the local reconstruction error of the history operator $R(t_d)$ satisfies:
  \begin{equation}
    \|x(t_d) - R(t_d)\| \le \frac{M_{K+1}}{(K+1)!} h_{\max}^{K+1},
  \end{equation}
  where $M_{K+1} = \sup_{t} \|x^{(K+1)}(t)\|$ over the local integration domain.
\end{lemma}

\begin{proof}
  Because the history reconstruction operator $R(t)$ evaluates a previously finalized polynomial segment exactly as defined in Lemma~\ref{lem:local_consistency}, this bound evaluates the localized Taylor remainder formula at the historical time point $t_d$. Like the assembled dense state, this strictly bounds the local truncation error of the history representation prior to the accumulation of global error.
\end{proof}

\begin{remark}[Continuity Across History Segments]
  Adjacent Taylor segments are generated sequentially from continuous state values enforced at event times, guaranteeing that the global history reconstruction operator $R(t)$ remains $C^0$-continuous across all segment boundaries. Consequently, when a dynamic or state-dependent delay forces the historical argument $t_d$ to cross a segment boundary, the operator smoothly transitions between independent polynomial pieces. Because each piece satisfies Lemma~\ref{lem:local_reconstruction_error} natively, the reconstruction error remains bounded uniformly by $\mathcal{O}(h_{\max}^{K+1})$, avoiding the localized interpolation mismatches or artificial order reduction characteristic of static grid-allocation schemes.
\end{remark}

\subsection{Delay Perturbation Analysis}

Define the error
$$e(t):=x_{num}(t)-x(t).$$
Let
$$t_d=t-\tau(t,x(t)), \qquad \hat t_d=t-\tau(t,x_{num}(t)).$$

\begin{lemma}
  $$|t_d-\hat t_d| \le L_\tau\|e(t)\|.$$
\end{lemma}

\begin{proof}
  By definition,
  $$|t_d-\hat t_d| = |\tau(t,x_{num}(t))-\tau(t,x(t))|.$$
  Applying the Lipschitz condition on $\tau$ gives
  $$|t_d-\hat t_d| \le L_\tau\|x_{num}(t)-x(t)\| = L_\tau\|e(t)\|.$$
\end{proof}

\begin{lemma}
  If $\|x'(t)\|\le M_1$, then
  $$\|x(t_d)-x(\hat t_d)\| \le M_1L_\tau\|e(t)\|.$$
\end{lemma}

\begin{proof}
  By the mean-value theorem,
  $$\|x(t_d)-x(\hat t_d)\| \le M_1|t_d-\hat t_d|.$$
  Applying the previous lemma yields
  $$\|x(t_d)-x(\hat t_d)\| \le M_1L_\tau\|e(t)\|.$$
\end{proof}

\subsection{Error Propagation and Continuous Defect}

Define the global error $e(t):=x_{num}(t)-x(t)$. 

\begin{definition}[Piecewise Continuous Defect]
  Let $\hat t_d=t-\tau(t,x_{num}(t))$. Because the assembled numerical state $x_{num}(t)$ is constructed from piecewise Taylor segments, it is globally continuous ($C^0$) but its derivative $x_{num}'(t)$ possesses jump discontinuities at the event times $t_k$. Therefore, we define the continuous residual, or defect, almost everywhere (for $t \neq t_k$) as:
  $$\rho(t) := x_{num}'(t) - f(x_{num}(t), x_{num}(\hat t_d)).$$
  This function measures how exactly the analytical derivative of the assembled numerical polynomial satisfies the governing dynamics evaluated at the reconstructed numerical state on the open intervals between events.
\end{definition}

\begin{definition}[Compact $\epsilon$-Tube]
  \label{def:epsilon_tube}
  Let $\epsilon > 0$. The compact $\epsilon$-neighborhood around the exact continuous solution history is defined as:
  $$\Omega_\epsilon := \left\{ y \in \mathbb{R}^m : \min_{t \in [-\tau_{\max}, T]} \|y - x(t)\| \le \epsilon \right\}.$$
\end{definition}

\begin{lemma}[Piecewise Bounded Evolution Derivatives]
  \label{lem:bounded_F}
  Under Assumption 1, $f, \tau \in C^{K+1}$. Because the assembled numerical state $x_{num}(t)$ is constructed from piecewise polynomials of degree $K$, it is globally $C^0$ and its time derivatives up to order $K$ are piecewise continuous. 
  
  Crucially, the delayed argument $t_d(t) = t - \tau(t, x_{num}(t))$ maps current integration times back into the numerical history, potentially crossing the discrete boundaries where history segments transition. Under Assumption \ref{assum:causality} (Causality), the delayed argument is strictly increasing ($1 - \frac{d\tau}{dt} > 0$), forcing $t_d(t)$ to pass transversally through these segment transitions. Because the set of asynchronous event times is finite over the bounded interval $[0,T]$, the pre-image of these discrete boundaries under the strictly monotonic mapping $t_d(t)$ is also finite. Consequently, the set of times where the numerical trajectory evaluates exactly on a non-smooth historical boundary has Lebesgue measure zero.
  
  Therefore, the multivariate chain rule expansion of the $K$-th time derivative $F^{(K)}(t)$ exists and is piecewise continuous almost everywhere (specifically, avoiding local event updates and their delay-induced echoes).
  
  Let $\Omega_\epsilon$ be the compact $\epsilon$-neighborhood defined in Definition \ref{def:epsilon_tube}. Because $F^{(K)}(t)$ is a piecewise continuous mapping composed of $C^{K+1}$ functions evaluated over a bounded domain, its magnitude is bounded almost everywhere. We can thus explicitly define the finite constant $M_F(\epsilon) < \infty$ as the essential supremum of this derivative mapping:
  $$M_F(\epsilon) := \operatorname*{ess\,sup} \left\{ \|F^{(K)}(t)\| : x_{num}(s) \in \Omega_\epsilon \text{ for all } s \le t \right\}.$$
  
  Consequently, for any arbitrary bootstrap interval
  $[0, t^*] \subseteq [0,T]$ where the numerical trajectory satisfies
  $x_{num}(s) \in \Omega_\epsilon$ for all $s \le t^*$, the continuous evolution
  function\\ $F(t) = f(x_{num}(t), x_{num}(t-\tau(t,x_{num}(t))))$ strictly
  satisfies:
  $$\|F^{(K)}(t)\| \le M_F(\epsilon) < \infty \quad \text{a.e.}$$
\end{lemma}

\begin{lemma}[Component-Wise Exactness of Algorithmic Coefficient Generation]
  \label{lem:coeff_exactness}
  Let $t_k^{(i)}$ be the event time at which component $i$ is asynchronously updated. Because the solver maintains dense-output polynomials $p_j(t)$ for all components $j = 1, \dots, m$, the numerical state $x_{num}(t)$ is a fully defined, piecewise-polynomial vector function.
  
  Provided the delayed argument $t_d = t_k^{(i)} - \tau(t_k^{(i)}, x_{num}(t_k^{(i)}))$ lies in the interior of a stored history segment, the delay composition is differentiable through order $K$. Because historical segment boundaries form a set of measure zero, this condition holds almost everywhere. 
  
  Under this condition, when the Taylor arithmetic and forward-mode Automatic Differentiation pipeline executes for component $i$ at $t_k^{(i)}$, the generated coefficients for the local polynomial $p_i(t)$ exactly match the Taylor coefficients of the $i$-th component of the numerical composite function:
  \[
    F_i(t) = f_i(x_{num}(t), x_{num}(t-\tau(t,x_{num}(t))))
  \]
  evaluated strictly at $t = t_k^{(i)}$, up to the truncation order $K$ used by the method.
\end{lemma}

\begin{proof}
  This relies fundamentally on the architecture of asynchronous dense-output. The AD pipeline does not merely read the scalar values of neighboring components $j \neq i$; it reads their $K$-th degree polynomial representations $p_j(t)$. As long as the delayed evaluation avoids exact segment boundaries (which occurs a.e.), the historical argument accesses a purely smooth $C^\infty$ polynomial function. 
  
  Because the standard rules of algorithmic Taylor arithmetic exactly preserve univariate and multivariate series composition \cite{berz2003taylor, griewank2008evaluating, perez2019taylorintegration}, the AD engine evaluates the exact temporal derivatives of $F_i(t)$ at the specific local event time $t_k^{(i)}$. The extension to state-dependent delay systems preserves this exactness because retrieving historical state components maps strictly to evaluating an explicit history polynomial $P_{\text{hist}}$ via standard algebraic Cauchy multiplication.
\end{proof}

\begin{lemma}[Assembled Global Truncation Defect]
  \label{lem:local_defect}
  Assuming the exact algorithmic generation of local Taylor coefficients up to degree $K$, the assembled global continuous defect $\rho(t) := x_{num}'(t) - F(t)$ satisfies the strict bound:
  $$\|\rho(t)\| \le \frac{M_F(\epsilon)}{K!} h_{\max}^K,$$
  almost everywhere for any $t \in [0, T]$, provided the numerical trajectory remains strictly within the compact tube $\Omega_\epsilon$.
\end{lemma}

\begin{proof}
  We construct the global defect bound by analyzing the system component-wise. At any given continuous time $t \in (0, T)$ not exactly on an event boundary, every component $i \in \{1, \dots, m\}$ is governed by its most recently generated polynomial $p_i(t)$, updated at some local event time $t_k^{(i)} \le t$.
  
  The local continuous defect for component $i$ is defined as $\rho_i(t) := p_i'(t) - F_i(t)$. 
  Because $p_i(t)$ is a polynomial of degree $K$, its derivative $p_i'(t)$ is a polynomial of degree $K-1$. By Lemma \ref{lem:coeff_exactness}, $p_i'(t)$ exactly matches the Taylor expansion of $F_i(t)$ around $t_k^{(i)}$ up to degree $K-1$.
  
  Applying the standard Taylor remainder theorem for a $(K-1)$-degree polynomial strictly to the $i$-th component yields:
  $$|\rho_i(t)| \le \frac{\sup_{\xi} \|F_i^{(K)}(\xi)\|}{K!} (t-t_k^{(i)})^K.$$
  
  By definition of the asynchronous scheduler, the elapsed time since the local update is strictly bounded by the maximum global event interval: $(t-t_k^{(i)}) \le \Delta t_k^{(i)} \le h_{\max}$. Furthermore, because the assembled numerical trajectory $x_{num}(s)$ is confined to the compact tube $\Omega_\epsilon$, Lemma \ref{lem:bounded_F} guarantees that the continuous $K$-th derivative of the vector field is uniformly bounded almost everywhere by $M_F(\epsilon)$.
  
  Consequently, every individual component strictly satisfies:
  $$|\rho_i(t)| \le \frac{M_F(\epsilon)}{K!} h_{\max}^K \quad \text{a.e.}$$
  
  Because this strict $\mathcal{O}(h_{\max}^K)$ upper bound holds simultaneously for all components $i=1, \dots, m$ on their respective active asynchronous intervals, it assembles directly into the global vector defect norm:
  $$\|\rho(t)\| = \|x_{num}'(t) - F(t)\| \le \frac{M_F(\epsilon)}{K!} h_{\max}^K \quad \text{a.e.}$$
\end{proof}

From the definition of the continuous defect and the true analytical derivative $x'(t) = f(x(t), x(t_d))$, the continuous error dynamics exactly satisfy almost everywhere (for $t \neq t_k$):
$$e'(t) = x_{num}'(t) - x'(t) = f(x_{num}(t),x_{num}(\hat t_d)) - f(x(t),x(t_d)) + \rho(t).$$

To simplify notation in the subsequent error bounds, we explicitly define the local truncation defect constant parameterized by the tube radius $\epsilon$:
$$\beta(\epsilon) := \frac{M_F(\epsilon)}{K!}.$$

Using the Lipschitz property of $f$ and substituting the bound on the continuous defect, we can bound the piecewise derivative of the error:
$$\|e'(t)\| \le L_0\|e(t)\| + L_d\|x_{num}(\hat t_d)-x(t_d)\| + \beta(\epsilon) h_{\max}^K \quad \text{a.e.}$$

To handle the delayed state difference, we add and subtract the true state evaluated at the approximate delayed time, $x(\hat t_d)$, and apply the triangle inequality:
$$\|x_{num}(\hat t_d)-x(t_d)\| \le \|x_{num}(\hat t_d)-x(\hat t_d)\| + \|x(\hat t_d)-x(t_d)\|.$$

The first term on the right-hand side is exactly the global error evaluated at the approximate delayed time, $\|e(\hat t_d)\|$. Since $\hat t_d \le t$, we can bound this by the supremum of the error over the history up to time $t$. Let $\|e_t\| := \sup_{s \in [-\tau_{\max}, t]} \|e(s)\|$. Then,
$$\|x_{num}(\hat t_d)-x(\hat t_d)\| = \|e(\hat t_d)\| \le \|e_t\|.$$

Applying the delay perturbation lemma to the second term yields\\
$\|x(\hat t_d)-x(t_d)\| \le M_1L_\tau\|e(t)\|$. Substituting these bounds back
into the delayed state difference gives:
$$\|x_{num}(\hat t_d)-x(t_d)\| \le \|e_t\| + M_1L_\tau\|e(t)\|.$$

Substituting this into the differential inequality for $\|e'(t)\|$ and grouping the non-delayed terms yields:
$$\|e'(t)\| \le (L_0 + L_d M_1 L_\tau)\|e(t)\| + L_d \|e_t\| + \beta(\epsilon) h_{\max}^K \quad \text{a.e.}$$

Letting $\alpha = L_0 + L_d M_1 L_\tau$, we obtain the fundamental piecewise delay-differential inequality governing the error propagation:
$$\|e'(t)\| \le \alpha\|e(t)\| + L_d \|e_t\| + \beta(\epsilon) h_{\max}^K \quad \text{a.e.}$$

\subsection{Bounded Continuous Error Propagation}

To establish a global bound on the error, we employ a delay-differential extension of Gronwall's inequality. Because the defect bounding constant $M_F(\epsilon)$ requires the trajectory to remain inside the compact tube $\Omega_\epsilon$, the error propagation must first be established conditionally.

\begin{lemma}[Conditional Bounded Error Propagation]
  \label{lem:conditional_gronwall}
  Assume the numerical trajectory remains strictly inside the compact neighborhood $\Omega_\epsilon$ over an arbitrary bootstrap interval $[0, t^*] \subseteq [0,T]$. Then there exists a finite constant $C_T(\epsilon) > 0$ such that the error satisfies:
  $$\sup_{t\in[0,t^*]}\|e(t)\| \le C_T(\epsilon) h_{\max}^{K}.$$
\end{lemma}

\begin{proof}
  Since $\|e(t)\| \le \|e_t\|$, we can simplify the differential bound to:
  $$\|e'(t)\| \le (\alpha + L_d)\|e_t\| + \beta(\epsilon) h_{\max}^K \quad \text{a.e.}$$

  Because $e(t)$ is absolutely continuous (globally $C^0$ and piecewise $C^1$), its fundamental theorem of calculus holds over the integration domain. Integrating from $0$ to $t$ (for $t \le t^*$) and applying the triangle inequality yields:
  $$\|e(t)\| \le \|e(0)\| + \int_0^t \left( (\alpha + L_d)\|e_s\| + \beta(\epsilon) h_{\max}^K \right) ds.$$

  Because the right-hand side is a monotonically non-decreasing function of $t$, it strictly bounds the supremum of the left-hand side over the interval $[0, t]$. Assuming exact initial history conditions (i.e., $\|e(0)\| = 0$), we obtain an integral inequality for the supremum function:
  $$\|e_t\| \le \beta(\epsilon) t h_{\max}^K + \int_0^t (\alpha + L_d)\|e_s\| ds.$$

  Applying the standard continuous Gronwall's lemma \cite{gronwall1919note} to the scalar function $u(t) = \|e_t\|$ yields:
  $$\|e_t\| \le \beta(\epsilon) t h_{\max}^K \exp\left((\alpha + L_d)t\right).$$

  Evaluating this bound up to the maximal domain time $T$ yields an unconditionally safe upper bound for any $t^* \le T$:
  $$\sup_{t\in[0,t^*]}\|e(t)\| \le \left[ \beta(\epsilon) T \exp\left((\alpha + L_d)T\right) \right] h_{\max}^K.$$

  Defining the constant $C_T(\epsilon) := \beta(\epsilon) T \exp\left((\alpha + L_d)T\right)$ completes the proof.
\end{proof}

\subsection{Global Convergence}

\begin{theorem}[Global Convergence and A Priori Stability]
  \label{thm:global_convergence_app}
  Under Assumptions 1, 2, and 3, there exists an explicit step size threshold $h_0 > 0$ such that for all maximum event intervals $h_{\max} \le h_0$, the global numerical error satisfies:
  $$\sup_{t\in[0,T]} \|x(t)-x_{num}(t)\| \le C_T(\epsilon) h_{\max}^{K}.$$
\end{theorem}

\begin{proof}
  We establish this using a bootstrap argument (continuous induction). Choose an arbitrary tube radius $\epsilon > 0$ and let $\Omega_\epsilon$ be its corresponding compact neighborhood. 
  
  We explicitly select the step size threshold $h_0(\epsilon)$ such that the maximum theoretically accumulated error bounded by Lemma \ref{lem:conditional_gronwall} remains strictly less than the chosen tube radius $\epsilon$:
  $$h_0(\epsilon) := \left( \frac{\epsilon}{2 C_T(\epsilon)} \right)^{1/K}.$$
  
  Assume for the sake of contradiction that the numerical trajectory escapes the tube $\Omega_\epsilon$ before the final time $T$. Let $[0, t^*]$ be the maximal bootstrap interval, where the escape time $t^*$ is defined as:
  $$t^* := \sup \{ t \in [0,T] : \|x(s) - x_{num}(s)\| \le \epsilon \text{ for all } s \in [0, t] \}.$$
  
  Because the solver is initialized with exact history conditions ($\|e(0)\| = 0 < \epsilon$) and the numerical trajectory is piecewise continuous, $t^* > 0$.
  
  For the entire bootstrap interval $[0, t^*)$, the numerical trajectory remains strictly in the interior of $\Omega_\epsilon$. Thus, by topological compactness, Lemma \ref{lem:bounded_F} ensures $M_F(\epsilon)$ remains finite, which strictly validates the delay-differential Gronwall error propagation in Lemma \ref{lem:conditional_gronwall}. By continuity, exactly at the escape time $t^*$, the error must satisfy $\|e(t^*)\| = \epsilon$. 
  
  However, because the solver operates with a maximum step size $h_{\max} \le h_0(\epsilon)$, applying the conditional Gronwall bound over the validated bootstrap interval $[0, t^*]$ guarantees:
  $$\|e(t^*)\| \le C_T(\epsilon) h_{\max}^K \le C_T(\epsilon) h_0(\epsilon)^K = \frac{\epsilon}{2} < \epsilon.$$
  
  This is a strict mathematical contradiction. The continuous error cannot simultaneously equal $\epsilon$ and be bounded by $\epsilon/2$. Therefore, our assumption of escape is false, and $t^* = T$. 
  
  Consequently, the numerical trajectory never escapes the compact tube $\Omega_\epsilon$ during the integration domain $[0,T]$. This closes the bootstrap loop, satisfying the condition for Lemma \ref{lem:conditional_gronwall} globally, yielding the final convergence bound:
  $$\sup_{t\in[0,T]} \|x(t)-x_{num}(t)\| \le C_T(\epsilon) h_{\max}^K.$$
\end{proof}

\begin{corollary}[Uniform Convergence]
  $$\sup_{t\in[0,T]} \|x(t)-x_{num}(t)\| = O(h_{\max}^{K}).$$
\end{corollary}

\subsection{Adaptive Error Control}

Unlike embedded Runge-Kutta methods that rely on carefully tuned coefficient tableaus to estimate local error, Taylor-series methods natively possess an intrinsic error estimator. By evaluating the magnitude of the highest-order computed term (or the difference between a $K$-th and $(K-1)$-th degree expansion), the solver directly measures the local truncation error. Provided the $(K+1)$-th derivative of the state trajectory does not entirely vanish, this algorithmic estimator strictly mirrors the theoretical asymptotic scaling. We formalize this as an operational assumption for the adaptive controller:

\begin{assumption}[Asymptotically Exact Error Estimator]
  \label{assum:estimator}
  The asynchronous adaptive controller evaluates event intervals using a local truncation error estimator $\eta(h)$ that perfectly captures the theoretical polynomial scaling. There exist positive constants $c_1, c_2 > 0$ such that for any sufficiently small event interval $h$:
  $$c_1 h^{K+1} \le \eta(h) \le c_2 h^{K+1}.$$
\end{assumption}

\begin{theorem}[Adaptive Convergence Order]
  \label{thm:adaptive_convergence}
  Under Assumption \ref{assum:estimator}, if the controller strictly enforces the local tolerance bound $\eta(h) \le \varepsilon$ across all asynchronous components, then the global numerical error satisfies:
  $$\|x(T) - x_{num}(T)\| = O(\varepsilon^{K/(K+1)}).$$
\end{theorem}

\begin{proof}
  By Assumption \ref{assum:estimator}, any accepted step size $h$ must satisfy the local tolerance constraint $\eta(h) \le \varepsilon$.
  
  Applying the lower bound of the estimator's asymptotic behavior, we have:
  $$c_1 h^{K+1} \le \eta(h) \le \varepsilon.$$
  
  Rearranging this inequality to isolate the step size $h$ yields:
  $$h \le \left( \frac{\varepsilon}{c_1} \right)^{\frac{1}{K+1}}.$$
  
  Because this bound applies universally to all accepted steps across all asynchronous coordinates, the maximum global event interval $h_{\max}$ is strictly bounded by:
  $$h_{\max} = O(\varepsilon^{1/(K+1)}).$$
  
  From the Unconditional Global Convergence theorem, we have already rigorously established that the global error is bounded by $C_T(\epsilon) h_{\max}^K$. Substituting our bound for $h_{\max}$ into this global error bound yields:
  $$\|x(T) - x_{num}(T)\| \le C_T(\epsilon) \left( \left( \frac{\varepsilon}{c_1} \right)^{\frac{1}{K+1}} \right)^K = \frac{C_T(\epsilon)}{c_1^{K/(K+1)}} \varepsilon^{K/(K+1)}.$$
  
  Because $C_T(\epsilon)$ and $c_1$ are constants independent of $\varepsilon$, it immediately follows that the global error scales strictly as $O(\varepsilon^{K/(K+1)})$.
\end{proof}

\end{document}